\documentclass[11pt,a4paper]{amsart}
\usepackage{fullpage}
\usepackage[OT1]{fontenc}
\usepackage[OT1,euler-digits]{eulervm}
\usepackage{microtype}
\usepackage[tt=false]{libertine} 
\usepackage{bm}
\usepackage{amsmath}
\usepackage{amsthm}
\usepackage{mathtools} 
\usepackage{algorithm}
\usepackage{algpseudocode}
\usepackage{enumitem}

\usepackage{tikz}
\usetikzlibrary{arrows,arrows.meta,backgrounds,calc,fit,decorations.pathreplacing,decorations.markings,shapes.geometric}

\def\scale{0.6}

\tikzstyle{internal} = [draw, fill, shape=circle]
\tikzstyle{external} = [shape=circle]
\tikzstyle{square}   = [draw, fill, rectangle]
\tikzstyle{triangle} = [draw, fill, regular polygon, regular polygon sides=3, inner sep=3pt]
\tikzstyle{pentagon} = [draw, fill, regular polygon, regular polygon sides=5, inner sep=2pt, minimum size=14pt]

\usepackage{pgfplots}

\usepackage{caption} 
\usepackage{subfig}

\usepackage[thinlines]{easytable}

\tikzset{every fit/.append style=text badly centered}

\usepackage[bookmarks=true,hypertexnames=false,pagebackref]{hyperref}
\hypersetup{colorlinks=true, citecolor=blue, linkcolor=red, urlcolor=blue}

\usepackage{ifthen}

\usepackage{cleveref}

\usepackage[textsize=tiny]{todonotes}
\usepackage{mleftright}
\usetikzlibrary{positioning,chains,fit,shapes,calc}
\usetikzlibrary{trees}
\usetikzlibrary{decorations.pathreplacing}
\usetikzlibrary{decorations.pathmorphing}
\usetikzlibrary{decorations.markings}
\tikzset{>=latex} 
\tikzset{
    reverseclip/.style={insert path={(current page.north east) --
        (current page.south east) --
        (current page.south west) --
        (current page.north west) --
        (current page.north east)}
    }
}

\usepackage{cool}
\Style{DSymb={\mathrm d},DShorten=true,IntegrateDifferentialDSymb=\mathrm{d}}

\newcommand{\tp}[1]{{\left( #1 \right)}}

\newcommand{\numP}{\#{\bf P}}

\renewcommand{\mathbf}[1]{\bm{#1}}

\def\*#1{\mathbf{#1}}
\def\+#1{\mathcal{#1}}
\def\-#1{\mathrm{#1}}
\def\=#1{\mathbb{#1}}

\newcommand{\abs}[1]{\left\vert#1\right\vert}

\newcommand{\set}[1]{\left\{#1\right\}}
\newcommand{\eps}{\varepsilon}

\renewcommand{\vec}[1]{\bm{#1}}

\newcommand{\trans}[4]{\ensuremath{\left[\begin{smallmatrix} #1 & #2 \\ #3 & #4 \end{smallmatrix}\right]}}

\newtheorem{theorem}{Theorem}
\newtheorem{conjecture}[theorem]{Conjecture}
\newtheorem{lemma}[theorem]{Lemma}

\newtheorem{proposition}[theorem]{Proposition}

\crefname{theorem}{Theorem}{Theorems}
\crefname{observation}{Observation}{Observations}
\crefname{claim}{Claim}{Claims}
\crefname{condition}{Condition}{Conditions}
\crefname{algorithm}{Algorithm}{Algorithms}
\crefname{property}{Property}{Properties}
\crefname{example}{Example}{Examples}
\crefname{fact}{Fact}{Facts}
\crefname{lemma}{Lemma}{Lemmas}
\crefname{corollary}{Corollary}{Corollaries}
\crefname{definition}{Definition}{Definitions}
\crefname{remark}{Remark}{Remarks}
\crefname{proposition}{Proposition}{Propositions}
\crefname{equation}{equation}{equations}
\crefname{enumi}{Case}{Case}
\creflabelformat{enumi}{(#2#1#3)}

\makeatletter
\def\prob#1#2#3{\goodbreak\begin{list}{}{\labelwidth\z@ \itemindent-\leftmargin
      \itemsep\z@  \topsep6\p@\@plus6\p@
      \let\makelabel\descriptionlabel}
  \item[\textbf{Name}]#1
  \item[\textbf{Instance}]#2
  \item[\textbf{Output}]#3
  \end{list}}
\makeatother

\newcommand{\defeq}{:=}
\newcommand{\prodK}[1]{\ensuremath{\+K_{#1}}}

\newcommand{\lambdastar}[1]{\ensuremath{\lambda^{\star}_{#1}}}
\def\lambdac{\ensuremath{\lambda_c}}
\def\dc{\ensuremath{d_c}}
\def\lambdaMCMC{\ensuremath{\lambda_{\text{MCMC}}}}
\def\dMCMC{\ensuremath{d_{\text{MCMC}}}}
\def\dstar{\ensuremath{d^{\star}}}
\def\cstar{\ensuremath{c^{\star}}}
\newcommand{\Zspin}{\ensuremath{Z_{\text{spin}}}}
\newcommand{\Complement}[1]{#1^{\textnormal{\texttt{c}}}}
\newcommand{\Closure}[1]{\overline{#1}}
\def\ii{\iota}
\def\ravg{\widehat{r}}
\def\zavg{\widehat{z}}
\def\thetaavg{\widehat{\theta}}
\def\DET{\Phi}
\newcommand{\alternative}[1]{\widetilde{#1}}
\def\righthalfplane{[0,\pi/2)\cup(3\pi/2,2\pi)}
\def\lefthalfplane{[\pi/2,3\pi/2]}

\title{Zeros of ferromagnetic 2-spin systems}

\author{Heng Guo}
\address[Heng Guo]{School of Informatics, University of Edinburgh, Informatics Forum, Edinburgh, EH8 9AB, United Kingdom.}
\email{hguo@inf.ed.ac.uk}

\author{Jingcheng Liu}
\address[Jingcheng Liu]{Department of EECS, University of California, Berkeley, CA}
\email{liuexp@berkeley.edu}

\author{Pinyan Lu}
\address[Pinyan Lu]{ITCS, Shanghai University of Finance and Economics, No.100
Wudong Road, Yangpu District, Shanghai, China.}
\email{lu.pinyan@mail.shufe.edu.cn}

\begin{document}
  
\begin{abstract}
We study zeros of the partition functions of ferromagnetic 2-state spin systems in terms of the external field,
  and obtain new zero-free regions of these systems via a refinement of Asano's and Ruelle's contraction method. 
  The strength of our results is that they do not depend on the maximum degree of the underlying graph. 
  Via Barvinok's method, we also obtain new efficient and deterministic approximate counting algorithms. 
  In certain regimes, our algorithm outperforms all other methods such as Markov chain Monte Carlo and correlation decay.
\end{abstract}

\maketitle

\section{Introduction}

Spin systems are widely studied in statistical physics, probability theory, machine learning, and theoretical computer science,
sometimes under a different name such as \emph{Markov random field}.
An important special case is when there are only $2$ spins, and
a systematic study of their computational complexity was initiated by Goldberg et al.~\cite{GJP03}.
In addition to their intrinsic importance, 
these systems are also great test beds for algorithmic ideas.
Many interesting tools and techniques are developed through studying them.
By now, we have almost completely settled the anti-ferromagnetic case,
whereas a definitive answer to the ferromagnetic case still remains elusive.

Before reviewing the state-of-the-art, we define the $2$-state spin system first.
In a graph $G=(V,E)$,
a configuration $\sigma:V \rightarrow \{0,1\}$ assigns one of the two spins ``0'' and ``1'' to each vertex.
The $2$-spin system is specified by the edge interaction matrix, which we normalise to $\trans{\beta}{1}{1}{\gamma}$,
and the external field $\lambda$ for vertices that are assigned $1$.
All parameters here are non-negative.
For a particular configuration $\sigma$,
its weight $w(\sigma)$ is a product over all edge interactions and vertex weights, that is
\begin{align} \label{eqn:weight}
  w(\sigma) = \beta^{m_0(\sigma)} \gamma ^{m_1(\sigma)} \lambda^{n_1(\sigma)},
\end{align}
where $m_0(\sigma)$ is the number of $(0,0)$ edges given by the configuration $\sigma$, $m_1(\sigma)$ is the number of $(1,1)$ edges,
and $n_1(\sigma)$ is the number of vertices assigned $1$.
The Gibbs measure / distribution of the system is one where the probability of a configuration is proportional to its weight.
The partition function \Zspin\ is the normalising factor of the Gibbs distribution:
\begin{align}  \label{eqn:Z}
  \Zspin(G;\beta,\gamma,\lambda)=\sum_{\sigma:V \rightarrow \{0,1\}} w(\sigma) = \sum_{\sigma:V \rightarrow \{0,1\}}\beta^{m_0(\sigma)} \gamma ^{m_1(\sigma)} \lambda^{n_1(\sigma)}.
\end{align}
An important special case is the Ising model, where $\beta=\gamma$.
Notice that in the statistical physics literature,
parameters are usually chosen to be the logarithms of our parameters above.
Change of variables as such do not affect the complexity of the same system.

Many macroscopic properties of the system can be studied through partition functions,
which raises the interest of computing them.
Exact computation of \Zspin\ is \numP-hard for all but trivial cases \cite{Bar82},
so the main focus is on approximating \Zspin.

The system shows drastically different behaviours depending on whether $\beta\gamma<1$ or  $\beta\gamma>1$ (the case where $\beta\gamma=1$ is degenerate).
The antiferromagnetic case $\beta\gamma<1$ is now very well understood by a series of work \cite{Wei06,LLY13,SST14,SS14,GSV16}, 
where an exact threshold of computational complexity transition is identified and the only remaining case is at the critical point.
This threshold corresponds to the uniqueness threshold of Gibbs measures in infinite regular trees (also known as the Bethe lattice).

On the other hand, far less is known for the ferromagnetic case $\beta\gamma>1$.
Due to symmetry, we will assume $\beta\ge\gamma$ throughout this paper as the other case is similar.
This assumption means that the edge interaction favours the spin ``0''.
As it turns out, if the external field also favors ``0'' (namely $\lambda\le 1$),
then efficient algorithms can be obtained in a number of ways.
The real challenge is how far we can allow $\lambda$ to go beyond $1$,
and a critical threshold is conjectured to exist.

Unlike antiferromagnetic systems, the tree uniqueness threshold is not the right answer,
as the pioneering algorithm of Jerrum and Sinclair \cite{JS93} is efficient on both sides of the tree uniqueness threshold for ferromagnetic Ising models ($\beta=\gamma$).
This algorithm is based on the Markov chain Monte Carlo (MCMC) method.
The MCMC method has been adapted to general ferromagnetic $2$-spin systems \cite{GJP03}.
The bound in \cite{GJP03} is then slightly improved \cite{LLZ14} to give an efficient approximation algorithm of \Zspin\ if $0<\lambda\le\lambda_{\text{MCMC}}=\frac{\beta}{\gamma}$,
for fixed $\beta\ge\gamma$.

The algorithmic success in the anti-ferromagnetic case is largely thanks to the correlation decay method introduced by Weitz \cite{Wei06}.
It is natural to try this method on ferromagnetic systems as well.
Non-trivial results have been obtained \cite{GL18} but these results still fall short from solving the problem in general.
In \cite{GL18}, the first and the third author raised the following conjecture.

\begin{conjecture}[\cite{GL18}]\label{conj:ferro}
  Let $\beta,\gamma,\lambda$ be positive parameters such that $\beta\ge\gamma$ and $\beta\gamma>1$.
  If \;$\lambda \le \lambdac$ where $\lambdac \defeq \left(\frac{\beta}{\gamma}\right)^{\dc}$ and $\dc\defeq\frac{\sqrt{\beta\gamma}}{\sqrt{\beta\gamma}-1}$,
  then a fully polynomial-time approximation scheme (FPTAS) exists for \Zspin.
\end{conjecture}

\Cref{conj:ferro} is confirmed in \cite{GL18} for the case of $\gamma\le1$.
However, it does not generalise to $\gamma>1$ because certain key properties in correlation decay fail.
On the other hand, one should not expect to go beyond $\lambdac$ too far.
Indeed, Liu et al.~\cite{LLZ14} identified another threshold beyond which the problem is as hard as approximately counting independent set in bipartite graphs,
which is a notorious open problem in approximate counting and is conjectured to have no efficient algorithm \cite{DGGJ04}.
This hardness threshold of \cite{LLZ14} is almost equal to $\lambdac$ except for a small integral gap.

In this paper, we obtain new algorithmic result that outperforms both the MCMC and the correlation decay methods in the $\gamma>1$ regime.
\begin{theorem}\label{thm:main}
  Let $\beta,\gamma,\lambda$ be positive parameters such that $\beta\ge\gamma$ and $\beta\gamma>1$.
  If $\lambda < \lambdastar{}$ where $\lambdastar{} \defeq \left(\frac{\beta}{\gamma}\right)^{\dstar/2}$ and $\dstar\defeq\frac{\pi}{\ArcTan {\sqrt{\beta\gamma-1}}}$,
  then an FPTAS exists for \Zspin\ in bounded degree graphs.
\end{theorem}
\Cref{thm:main} is a generalisation of the algorithm for the ferromagnetic Ising model ($\beta=\gamma$) by Liu, Sinclair, and Srivastava~\cite{LSS17}.
We note that our bound on $\lambda$ is uniform and does not rely on the maximum degree of the underlying graph.
The requirement of bounded degree is only for the efficiency of our algorithm.
Without this assumption our algorithm becomes quasi-polynomial time.
This is typical for deterministic approximate counting algorithms.

To compare $\lambdastar{}$ with $\lambdac$, we note that as $\beta\gamma\rightarrow 1$, $\dstar$ is asymptotically the square root of $\dc$.
An illustration of comparing $\lambdaMCMC$, $\lambdac$ and $\lambdastar{}$ is given in \Cref{fig:main}.

\pgfmathdeclarefunction{dstar}{1}{%
  \pgfmathparse{acos(#1)}%
}

\begin{figure}[htbp]
  \centering
  \begin{minipage}{.4\textwidth}
    \centering
      \begin{tikzpicture}[scale=\scale,transform shape]
        \begin{axis}[	xmin=1/2,   xmax=2, ymin=0,   ymax=20, extra x ticks={1}, extra y ticks={1}, extra tick style={grid=major},xlabel={\Large $\gamma$},ylabel={Threshold of the external field},legend entries={$\lambdac$,$\lambdastar{}$,$\lambdaMCMC$}]
        \addplot [very thick, draw=red,  domain=0.8:1, smooth] {(2/x)^( sqrt(2*x)/(sqrt(2*x)-1))};
        \addplot [very thick, draw=blue,  domain=1/2:2, smooth] {(2/x)^(90/acos(1/sqrt(2*x)))}; 
        \addplot [very thick, draw=orange,  domain=1/2:2, smooth] {2/x};
        \addplot [very thick, draw=red,  domain=1:2, dashed, smooth] {(2/x)^( sqrt(2*x)/(sqrt(2*x)-1))};
      \end{axis}
      \end{tikzpicture}
  \end{minipage}
  \hspace{1em}
  \begin{minipage}{.4\textwidth}
    \centering
      \begin{tikzpicture}[scale=\scale,transform shape]
        \begin{axis}[xmin=1/2,   xmax=2, ymin=0, ymax=6, extra x ticks={1}, extra y ticks={1}, extra tick style={grid=major},xlabel={\Large $\gamma$},ylabel={Exponent of the threshold},legend entries={$\dc$,$\dstar/2$,$\dMCMC$}]
        \addplot [very thick, draw=red,  domain=0.6:1, smooth] {sqrt(2*x)/(sqrt(2*x)-1)};
        \addplot [very thick, draw=blue,  domain=1/2:2, samples=100, smooth] {90/acos(1/sqrt(2*x))}; 
        \addplot [very thick, draw=orange,  domain=1/2:2, smooth] {1};
        \addplot [very thick, draw=red,  domain=1:2, dashed, smooth] {sqrt(2*x)/(sqrt(2*x)-1)};        
      \end{axis}
      \end{tikzpicture}
  \end{minipage}
  \caption{ 
    Fix $\beta=2$ and the range of $\gamma$ is $(1/2,2]$.
    Left: comparison of $\lambdaMCMC$, $\lambdac$ and $\lambdastar{}$.
    Right: comparison of $\dMCMC\defeq 1$, $\dc$ and $\dstar/2$.
    The dashed red line marks the conjectured threshold for $\gamma\ge 1$.}
  \label{fig:main}
\end{figure}
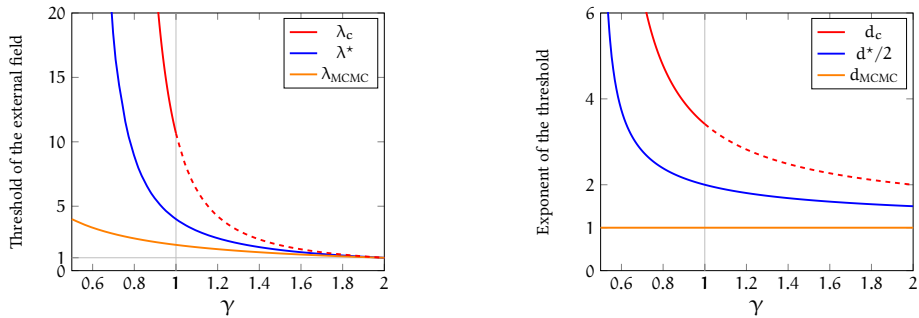

Our algorithm is based on a recent algorithmic technique developed by Barvinok \cite{Bar16} and extended by Patel and Regts \cite{PR17a}.
The idea is to view \Zspin\ as a polynomial in $\lambda$,
and turn zero-free regions of this polynomial in the complex plane into efficient approximation algorithms of the corresponding parameters.
The major challenge of applying this algorithmic framework is to obtain sharp zero-free regions along the real axis.

There are two main methods in obtaining zero-free regions.
The first one is the recursion method, 
where one gradually eliminates vertices from the graph, 
and shows that the zeros are always outside of the desired region.
This method has found many successes, see e.g.\ some work of Sokal \cite{Sok01,SS05}.
More recently, it has been successfully applied to solve long-standing conjectures \cite{PR19} and open problems \cite{LSS19}.
However, there are also strong connections between correlation decay and the recursion method.
In some sense, both results of \cite{PR19} and \cite{LSS19} are turning correlation decay analysis into zero-freeness bounds using complex dynamical systems.
For ferromagnetic 2-spin systems, 
because correlation decay fails if $\gamma>1$ \cite{GL18},
it would be surprising to obtain any meaningful result using the recursion method in this case.

In order to bypass the correlation non-decay barrier, 
we employed the other method, namely the contraction method, pioneered by Asano \cite{Asa70} and Ruelle \cite{Rue71,Rue99a}.
In a typical application, one starts with a graph of isolated components, and then contract vertices or edges to form the desired graph $G$.
The zero-free regions of isolated components are easy to analyse, but the contractions will spread the zeros across the complex plane.
The main effort is to control this spread.
In all previous applications of this method that we are aware of, 
either the unit circle or half planes are used as the starting point.
Our idea is to consider circles whose center and radius are carefully chosen (depending on the parameters), and sometimes their complements.
The main technical challenge is a detailed analysis for contracting an \emph{arbitrary} number of corresponding regions,
which involves repeated Minkowski product of circular regions.
We do so by solving a highly non-trivial optimisation problem in complex variables (see \eqref{eqn:program}).
It remains to be explored whether this methodology has other applications as well.

\begin{theorem}  \label{thm:zero-free}
  Let $\beta,\gamma$ be positive parameters such that $\beta\ge\gamma$ and $\beta\gamma>1$,
  and $\lambdastar{}$ defined as in \Cref{thm:main}.
  Then for any graph with minimum degree at least $2$,
  \Zspin, viewed as a polynomial in $\lambda$, is zero-free in a constant-sized small neighbourhood of the interval $[0,\lambda']$ for any $\lambda'<\lambdastar{}$.
\end{theorem}
The minimum degree requirement in \Cref{thm:zero-free} comes from some technical difficulty with degree $1$ vertices.
They do not affect the algorithmic result, \Cref{thm:main},
because we can preprocess the graph to remove the leaves, and then deal with an instance with non-uniform external fields.
In order to do so, we in fact show a stronger multivariate zero-free theorem, see \Cref{thm:zero-free-multi}.

The main message of our paper is to show that the failure of correlation decay is \emph{not} an essential barrier to obtain efficient algorithms.
However, because of some inherent difficulties of the contraction method, as explained in \Cref{sec:limit},
our result still falls short of confirming \Cref{conj:ferro}.
By now we have three different point of views for approximating \Zspin, namely MCMC, correlation decay, and zeros of polynomials.
They are just different aspects of the same object, 
and perhaps settling the complexity of ferromagnetic 2-spin systems requires a more unified view.

\section{Barvinok's algorithm}

Recall \eqref{eqn:weight} and \eqref{eqn:Z} that
\begin{align*}
  \Zspin(G;\beta,\gamma,\lambda)=\sum_{\sigma:V \rightarrow \{0,1\}} w(\sigma) = \sum_{\sigma:V \rightarrow \{0,1\}}\beta^{m_0(\sigma)} \gamma ^{m_1(\sigma)} \lambda^{n_1(\sigma)}.
\end{align*}
We will view \eqref{eqn:Z} as a polynomial in $\lambda$ and fix $\beta$ and $\gamma$.
In that case, we write $\Zspin(G;\lambda)$ for short.
The main effort of this paper is to show that for a certain region of $\lambda$ on the complex plane, $\Zspin\neq 0$.

Our interest in the zeros of the partition function is due to the algorithmic approach developed by Barvinok \cite[Section~2]{Bar16}. 
Let the $\delta$-strip of $[0,t]$ be
\begin{align*}
  \set{z\in\=C\mid \abs{\Im z}\le\delta\text{ and }-\delta\le\Re z\le t+\delta}.
\end{align*}
Suppose a polynomial $P(z)=\sum_{i=1}^nc_i z^i$ of degree $n$ is zero-free in a strip containing $[0,t]$.
Barvinok's method roughly states that $P(t)$ can be $(1\pm\eps)$-approximated using $c_0,\dots,c_k$ for some $k=O\tp{e^{\Theta\inp{1/\delta}} \cdot \log \frac{n }{\eps}}$,
via truncating the Taylor expansion of the logarithm of the polynomial.
In general, computing these coefficients naively will take quasipolynomial-time. 
However, Patel and Regts \cite{PR17a} have provided additional insights on how to compute these coefficients efficiently for a large family of graph polynomials in bounded degree graphs.
As explained in~\cite{LSS17}, the idea of Patel and Regts \cite{PR17a} can be applied to the partition functions of spin systems in much more generality, which includes $\Zspin(G;\lambda)$ that we are interested in. 
Thus, combining the algorithmic paradigm of Barvinok \cite[Section 2]{Bar16} and the idea of Patel and Regts \cite{PR17a}, we have the following useful lemma.

\begin{lemma} \label{lem:zeros-ferro-2-spin}
  Fix $\beta$, $\gamma$ and an integer $\Delta\ge 2$.
  Let $G$ be a graph of maximum degree $\Delta$.
  If $\Zspin(G;\lambda)$ does not vanish in a $\delta$-strip containing $[0,\lambda']$,
  then there is an FPTAS for $\Zspin(G;\lambda)$ for all $\lambda\in[0,\lambda']$.
\end{lemma}

In fact, as it has been observed in~\cite{PR17a}, the algorithm can be extended to a multivariate version of the partition fucntion easily. 
Let $\vec{ \lambda}\in \=C^{V}$ be a vector that specifies an external field for each vertex.
The multivariate partition function is given by
\begin{align}\label{eqn:Z-multi}
  \Zspin(G;\beta,\gamma,\vec{ \lambda})\defeq
  \sum_{\sigma:V \rightarrow \{0,1\}}\beta^{m_0(\sigma)} \gamma ^{m_1(\sigma)} \prod_{v \in V} \lambda_v^{[\sigma(v) = 1]}.
\end{align}

\begin{lemma} \label{lem:zeros-ferro-2-spin-multi}
  Fix $\beta$, $\gamma$ and an integer $\Delta\ge 2$.
  Let $G$ be a graph of maximum degree $\Delta$ and $n=\abs{V}$.
  If $\Zspin(G;\vec{\lambda})$ does not vanish in a $\delta$-polystrip 
  $ \set{\vec{z}\in\=C^n\mid \forall i\in [n], \abs{\Im z_i}\le\delta\text{ and }-\delta\le\Re z_i\le \lambda'+\delta}$,
  then there is an FPTAS for $\Zspin(G;\vec{\lambda})$ for all $\vec{\lambda}\in[0,\lambda']^n$.
\end{lemma}
\begin{proof}
	For any $\vec{\lambda}\in[0,\lambda']^n$, we consider the univariate polynomial $f(t) = \Zspin(G;t\cdot \vec{\lambda})$.
	On the one hand, $f(1)=\Zspin(G;\vec{\lambda})$ is the quantity what we want to approximate.
	On the other hand,  the fact that $\Zspin(G;\vec{\lambda})$ does not vanish in a $\delta$-polystrip containing the poly-region $[0,\lambda']^n$ 
	implies that there exists a $\delta'>0$ (depending on $\delta$ and $\lambda'$), such that $f(t)$ does not vanish in a $\delta'$-strip containing $[0,1]$.
    Hence, applying \Cref{lem:zeros-ferro-2-spin} on $f(t)$ yields our desired FPTAS for $\Zspin(G;\vec{\lambda})$.
\end{proof}
We note that for any fixed $\beta$, $\gamma$, $\lambda$, and $\Delta$,
our FPTAS runs in time bounded by a polynomial in $n=\abs{V}$ and $1/\varepsilon$.
However, as is typical for deterministic counting algorithms,
the exponent can grow with $\Delta$ and other parameters as they approach the threshold.

\section{The contraction method}

We use the contraction method to show zero-freeness for a $\delta$-strip containing part of the non-negative real line.
The contraction method is an important technique of bounding the zeros of graph polynomials \cite{Asa70, Rue71}.
It was first introduced by Asano \cite{Asa70} as an alternative way of proving the celebrated Lee-Yang circle theorem \cite{LY52}.

The contraction method has two main ingredients.
Firstly we want to relate zeros of a univariate polynomial with those of its polar form.
For a polynomial $P(z)=\sum_{i=0}^{d'} a_i z^i$ of degree $d'\le d$,
its $d$-th polar form with variables $\*z=(z_1,\dots,z_d)$ is
\begin{align*}
  \widehat{P}(\*z)\defeq\sum_{I\subseteq[d]}\frac{a_{\abs{I}}}{\binom{d}{\abs{I}}}z_I,
\end{align*}
where $a_i =0$ if $i> d'$, $[d]$ denotes $\{1,2,\dots,d\}$,
and for an index set $I$, $z_I=\prod_{i\in I}z_i$.
The polar form $\widehat{P}(\*z)$ is the unique multi-linear symmetric polynomial of degree at most $d'$ such that
$\widehat{P}(z,z,\dots,z) = P(z)$.
When $d'<d$, we view $P(z)$ as a degenerate case,
and it has zeros at $\infty$ with multiplicity $d-d'$.

Let $\+C$ be a region in $\=C$.
We say a polynomial $P(\*z)$ in $d\ge 1$ variables is \emph{$\+C$-stable} if $P(\*z)\neq 0$ whenever $z_1,\dots,z_d\in \+C$.
We call $\+C$ \emph{a circular region} if it is a disk, a half plane (a disk whose center is at infinity), or the complement of a disk in $\=C$.\footnote{Including complements of disks is slightly more general than what is usually stated, but this definition suits our purposes better and \Cref{prop:Grace-Walsh-Szego} is still true with this definition. See for example \cite[Section 3, Theorem 3.41b]{RS02}.}

The Grace-Szeg\H{o}-Walsh coincidence theorem \cite{Gra02,Sze22,Wal22} has the following immediate consequence.
\begin{proposition}  \label{prop:Grace-Walsh-Szego}
  If $\+C$ is a circular region,
  then a univariate polynomial $P(z)$ is $\+C$-stable if and only if its polar form $\widehat{P}(\*z)$ is $\+C$-stable.
\end{proposition}

The next ingredient is the Asano contraction \cite{Asa70,Rue71}.
We will use a slightly different version than the standard one.

\begin{lemma}  \label{lem:contraction}
  Let $K$ be a closed subsets of the complex plane $\=C$, which do not contain $0$,
  and $d\ge 1$ be an integer.
  If the complex polynomial
  \begin{align*}
    P(\*z):=\sum_{I\subseteq [d]} c_I \prod_{i\in I} z_i
  \end{align*}
  can vanish only when $z_i\in K$ for some $i\in[d]$,
  then
  \begin{align*}
    Q(z):=c_{\emptyset}+ c_{[d]} z
  \end{align*}
  can vanish only when $z\in \prodK{d} \defeq(-1)^{d+1} K \cdot K \cdots K$ ($d$ times).
\end{lemma}
\begin{proof}
  If $c_{[d]} = 0$, since $0\not\in K$, $P(0,0,\dots,0) = c_{\emptyset}\neq 0$.
  Thus, $Q(z)=c_{\emptyset}\neq 0$ for any $z$.

  Otherwise $c_{[d]}\neq 0$. 
  Consider the univariate polynomial $\widetilde{P}(z)=P(z,z,\dots,z)$ of degree $d$ and let $\zeta_1,\dots,\zeta_d$ be its roots.
  Clearly $\zeta_i\in K$ for all $i\in [d]$ because of the assumption.
  Thus, by Vieta's formula,
  \begin{align*}
    \prod_{i=1}^d \zeta_i = (-1)^d \frac{c_{\emptyset}}{c_{[d]}}.
  \end{align*}
  It implies that $-\frac{c_{\emptyset}}{c_{[d]}}\in \prodK{d}$.
\end{proof}

Some form of \Cref{lem:contraction} was first discovered by Asano \cite{Asa70} to provide a simple and alternative proof for the celebrated Lee-Yang circle theorem \cite{LY52},
where one chooses $K$ to be the unit disk or its complement.
The contraction method was further extended by Ruelle \cite{Rue71} and applied to subgraph counting polynomials \cite{Rue99a},
where one chooses $K$ to be half planes.
This choice has also found some algorithmic success recently \cite{GLLZ19}.
As we will see in the next section,
our choices are much more intricate,
including both disks and their complements,
and the center and radius are carefully calculated so that the result is optimal for the contraction method.

\section{Analyzing the contraction}

To apply \Cref{lem:contraction}, we will choose $K$ to be a closed circular region.
The specific choice will depend on the positivity of the following quantity:
\begin{align}\label{eqn:DET}
  \DET\defeq\log\sqrt{\frac{\beta}{\gamma}} - \ArcTan{\sqrt{\beta\gamma-1}}\cdot\sqrt{\beta\gamma-1}.
\end{align}
The main case is when $\DET<0$, which include the case of $\gamma>1$.
However, when $\gamma$ is sufficiently close to $1/\beta$, $\DET\ge0$ and we need a different solution.

\subsection{\texorpdfstring{$\DET<0$}{Phi < 0}}
\label{sec:phi<0}

In this case we choose the circular region to be the open disk centered at some real $c\ge0$ with radius $r>0$, denoted by $\+D(c,r)$.
Namely, $\+D(c,r)=\{z\in\=C\mid\abs{z-c} < r\}$.
Let $\+C(c,r)\defeq\partial\+D(c,r)$ be the circle centered at $c$ with radius $r$.
The region $K$ in \Cref{prop:Grace-Walsh-Szego} and \Cref{lem:contraction} will be chosen as the complement of the disk $\Complement{\+D(c,r)}$.
An illustration of $K$ and $\+K_i$ for $i=2,3,4$ is given in \Cref{fig:contraction}.

\begin{figure}[htpb]
  \centering
  \begin{minipage}{.2\linewidth}
    \includegraphics[width=\linewidth]{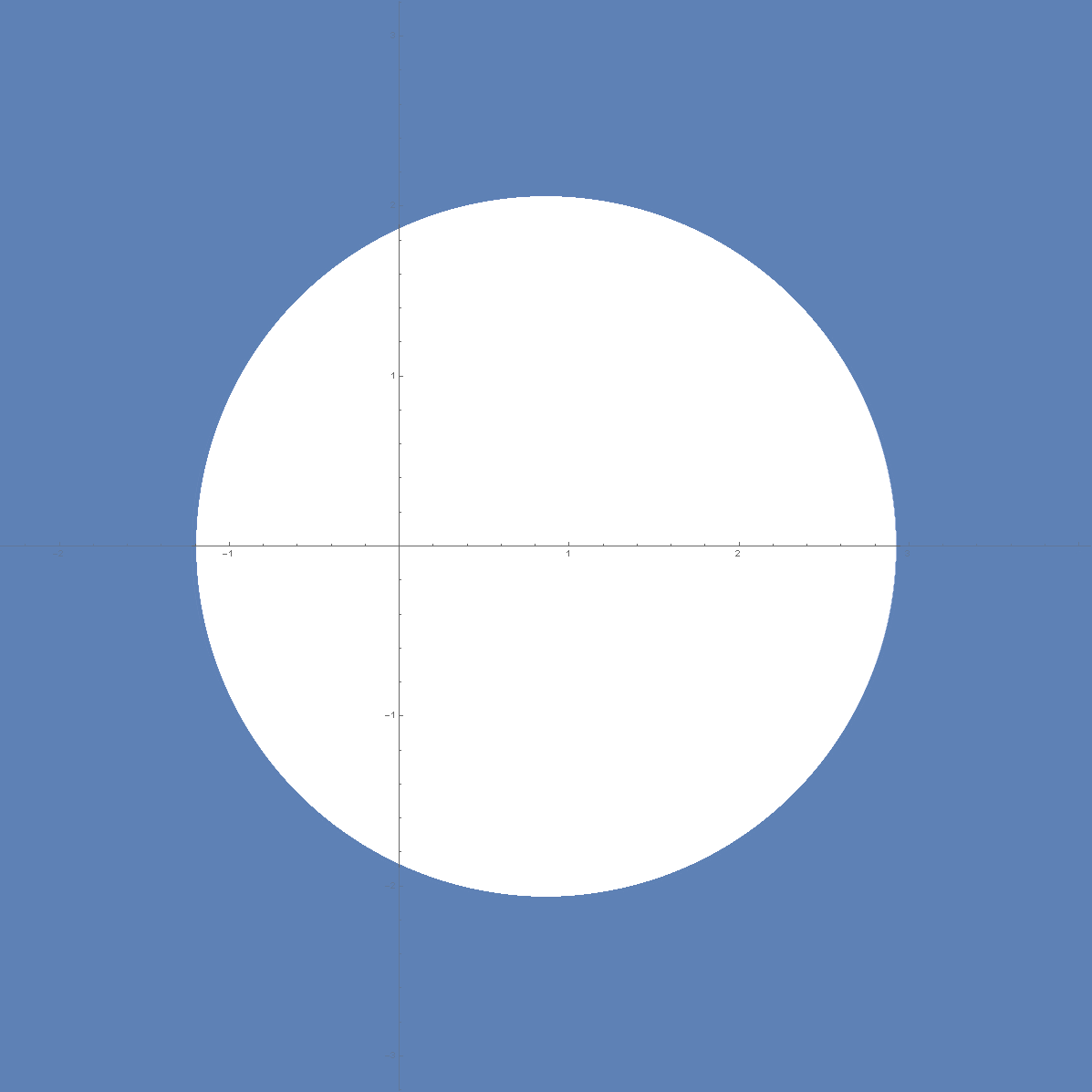}
    \captionsetup{width=0.9\linewidth}
    \captionof*{figure}{$K=\Complement{\+D(c,r)}$}
  \end{minipage}
  \hspace{0.5em}
  \begin{minipage}{.2\linewidth}
    \includegraphics[width=\linewidth]{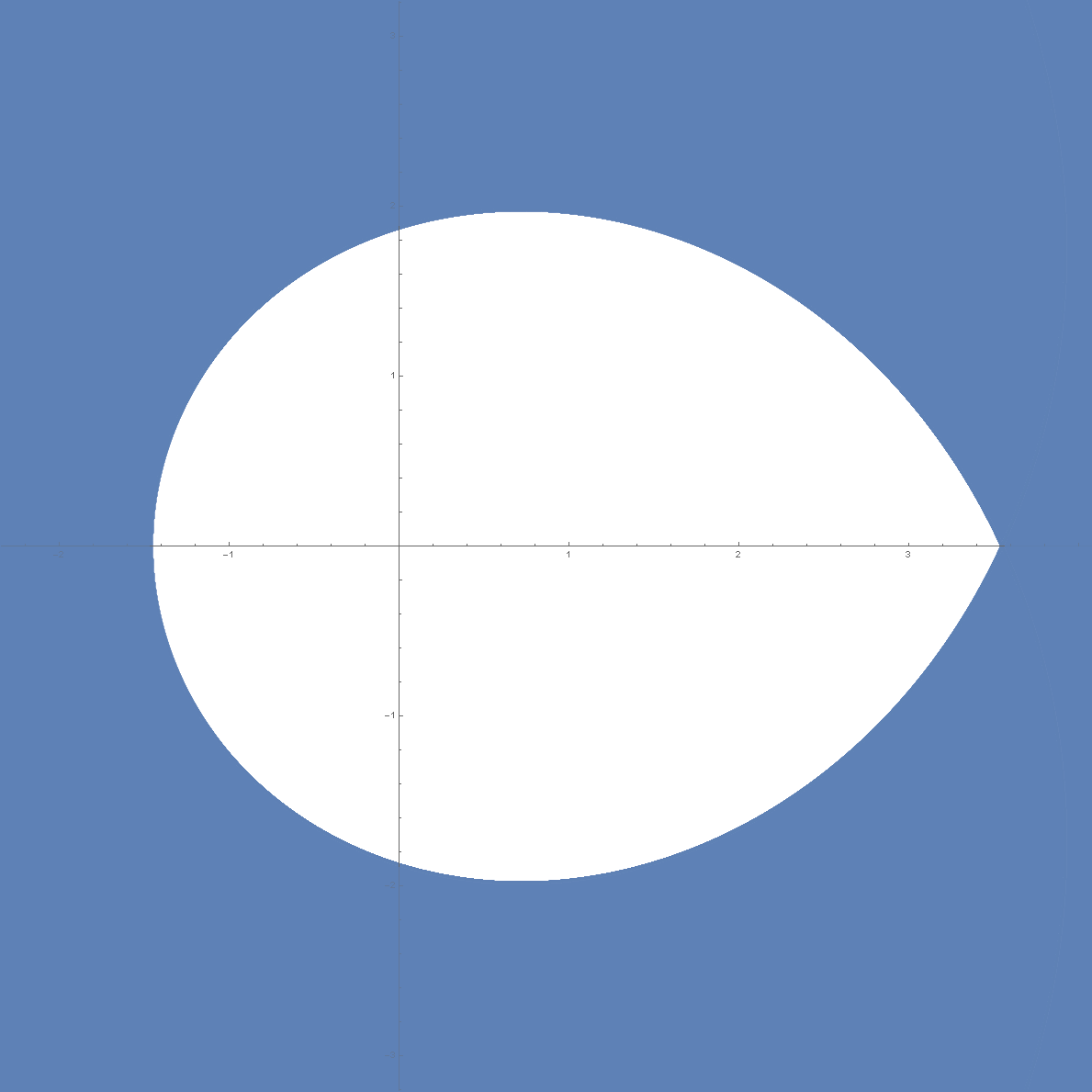}
    \captionsetup{width=0.9\linewidth}
    \captionof*{figure}{$\+K_2=-K\cdot K$}
  \end{minipage}
  \hspace{0.5em}
  \begin{minipage}{.2\linewidth}
    \includegraphics[width=\linewidth]{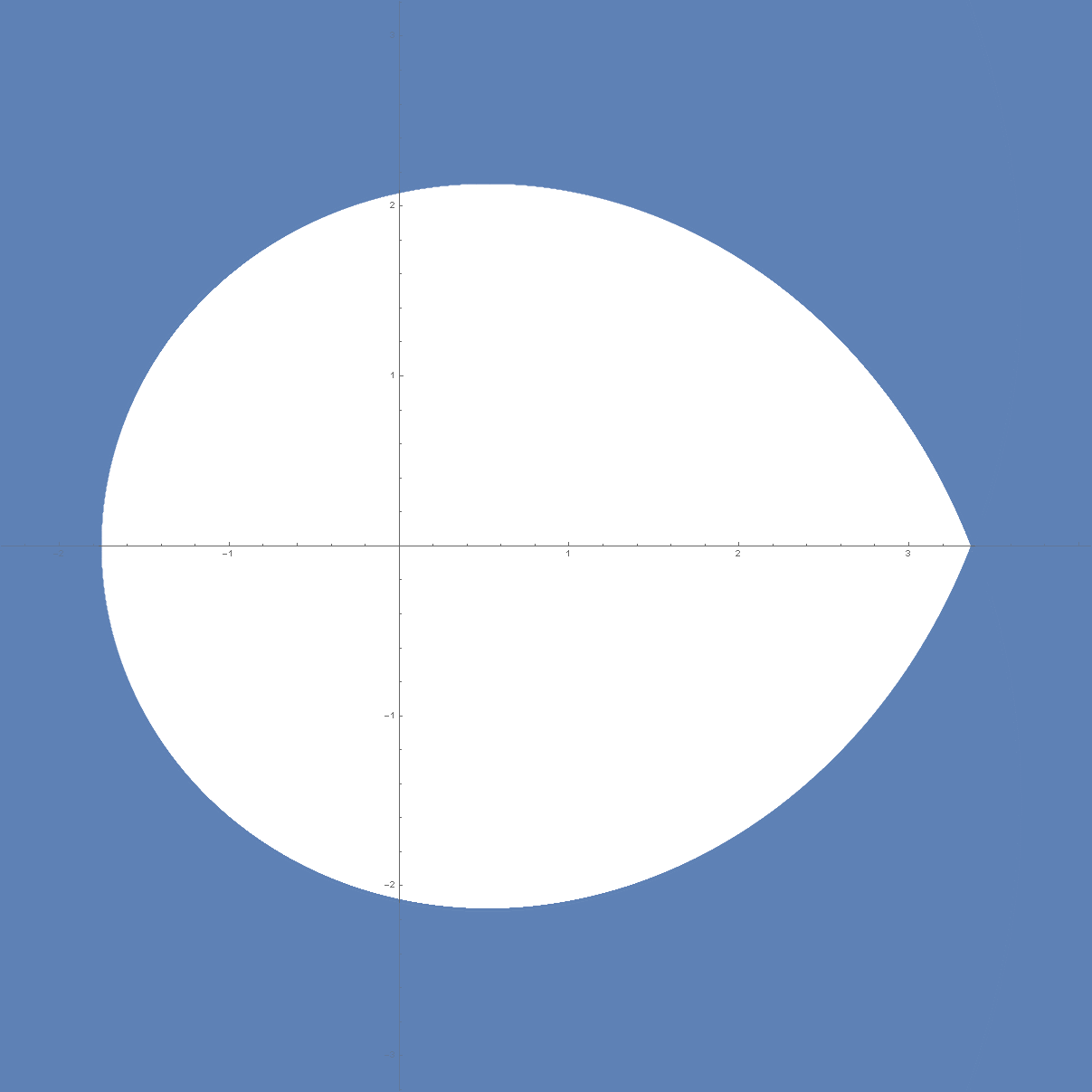}
    \captionsetup{width=0.9\linewidth}
    \captionof*{figure}{$\+K_3=K\cdot K\cdot K$}
  \end{minipage}
  \hspace{0.5em}
  \begin{minipage}{.2\linewidth}
    \includegraphics[width=\linewidth]{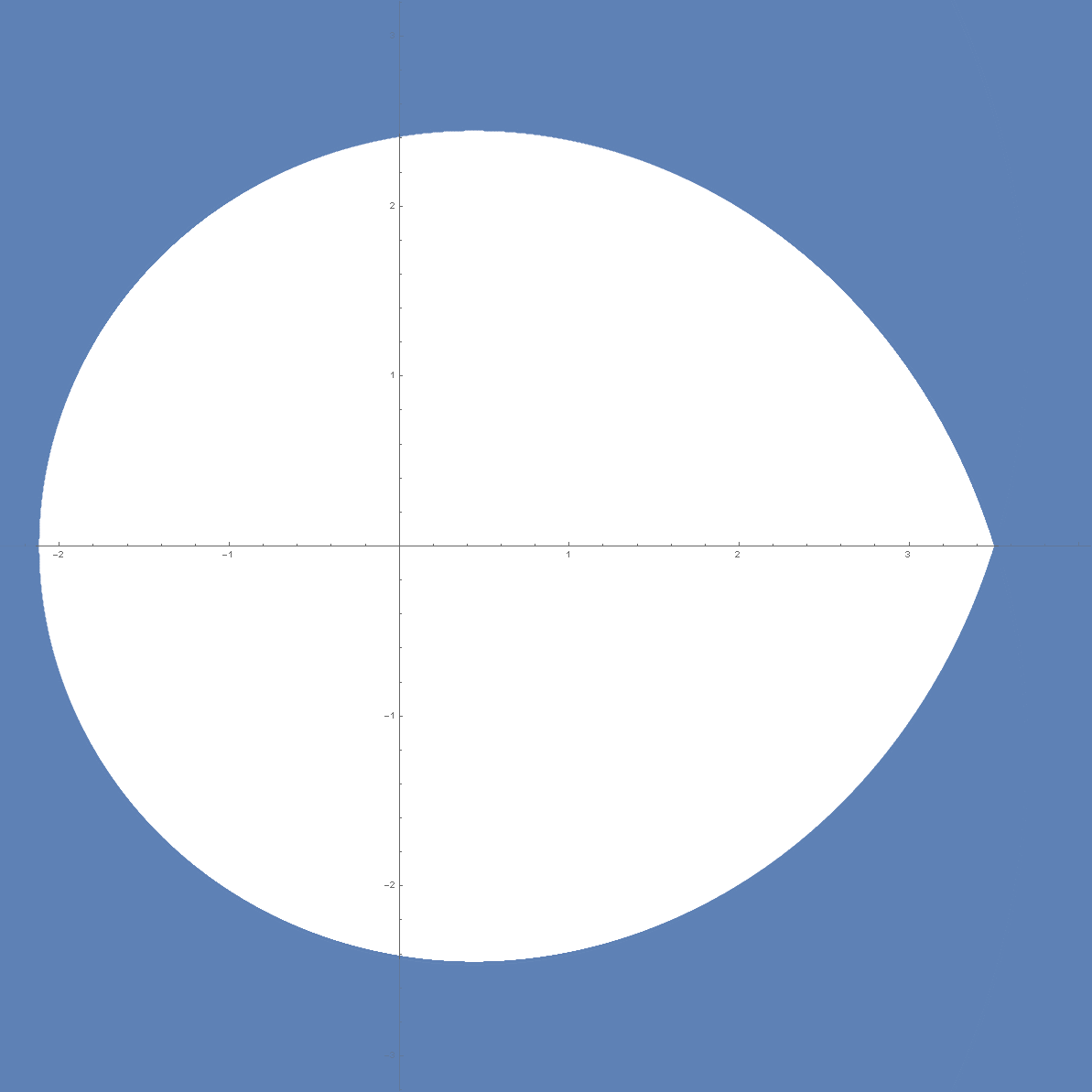}
    \captionsetup{width=0.9\linewidth}
    \captionof*{figure}{$\+K_4 = -K^4$}
  \end{minipage}
  \caption{Our region $K=\Complement{\+D(c,r)}$, $\+K_2$, $\+K_3$ and $\+K_4$ in the case of $\beta=3$ and $\gamma=\frac 4 3$.
  Here, the intercept of $\+K_d$ on the positive real line is minimised at $d=3$ for all $d\ge 2$.}
  \label{fig:contraction}
\end{figure}

When there is a single edge, the partition function is $\gamma \lambda^2+2\lambda+\beta$.
Due to the ferromagnetic assumption $\beta\gamma>1$, 
the equation $\gamma x^2+ 2x+\beta=0$ have two complex roots:
\begin{align}  \label{eqn:zeta}
  \zeta_1&=\frac{-1+\sqrt{1-\beta\gamma}}{\gamma}, & \zeta_2&=\frac{-1-\sqrt{1-\beta\gamma}}{\gamma}.
\end{align}
In particular $\abs{\zeta_1}=\abs{\zeta_2}=\sqrt{\frac{\beta}{\gamma}}$.
We will ensure that $\zeta_i$ lies on the boundary of the disk $\+D(c,r)$ we are choosing.
Namely, once $c$ is fixed, $r=r(c,\beta,\gamma)$ will be chosen to satisfy the following equation
\begin{align}  \label{eqn:r}
  \frac{\beta\gamma-1}{\gamma^{2}}+\left( c+\frac{1}{\gamma} \right)^2 = r^2.
\end{align}
Eventually, we will choose $c$ to be
\begin{align}  \label{eqn:c}
  \cstar\defeq\frac{-\beta\log\sqrt{\frac{\beta}{\gamma}}}{\DET}.
\end{align}
We remark that most of the argument in this subsection does not require $\Phi < 0$, but only requires that $0\le c<r$ is a positive real number.
The condition $\DET<0$ is only needed in the end, where we have to choose $c$. 

For some integer $d$, we want to argue that the complement of $\prodK{d}$ where $K=\Complement{\+D(c,r)}$ does not contain a neighbourhood of $[0,\lambda]$ for some $\lambda>0$.
Consider the following program:
\begin{align}
  \label{eqn:program}
  \begin{split}
  \min\quad & \prod_{i=1}^d r_i \\
  \text{subject to}\quad & \sum_{i=1}^d\theta_i = 
  \begin{cases}
    0 \mod 2\pi & \text{if $d$ is odd;}\\
    \pi \mod 2\pi & \text{if $d$ is even;}\\
  \end{cases}\\
  \quad & \forall i\in[d],\ r_i\ge 0 \text{ and } 0\le  \theta_i < 2\pi;\\
  \quad & \forall i\in[d],\ \abs{r_ie^{\ii\theta_i}-c}\ge r.    
  \end{split}
\end{align}
The last constraint ensures that $z_i\defeq r_ie^{\ii\theta_i}\in \Complement{\+D(c,r)}$ and the objective is to minimise the smallest positive real value in $\prodK{d}$.
An illustration is given in \Cref{fig:feasible}.

\def\cexample{0.863712}
\def\rexample{2.2}
\def\bounds{3}

\begin{figure}[h]
  \centering
  \begin{minipage}{.4\textwidth}
    \centering
    \begin{tikzpicture}[scale=1,transform shape]
        \begin{axis}[xtick=\empty, ytick=\empty, 
            trig format plots=rad, 
            axis lines = middle, axis line style={->}, axis on top,axis equal,
            axis background/.style={fill=blue!10!white},
            xmin=-1.5, xmax=3, ymin=-3, ymax=3,
            x label style={at={(axis cs:4,0)},anchor=north},
            y label style={at={(axis cs:0,3)},anchor=east},            
            xlabel = $\Re$, ylabel=$\Im$]
            \draw[thick, fill=white](axis cs:\cexample,0) circle [draw=black, radius=5.4em];  
          \node [label={270:{$c$}},circle,fill,inner sep=1pt] at (axis cs:{\cexample},0) {};
          \node [label={235:{$0$}},circle,fill,inner sep=1pt] at (axis cs:0,0) {};          
          \node [label={180:{$\zeta_1$}},circle,fill,inner sep=1pt] at (axis cs:-3/4,3/2) {};
          \node [label={180:{$\zeta_2$}},circle,fill,inner sep=1pt] at (axis cs:-3/4,-3/2) {};
          \node [label={90:{$z_i$}},circle,fill,inner sep=1pt] at (axis cs:1.96371,1.90526) {};
          \draw (axis cs:1.96371,1.90526) -- (axis cs:{\cexample},0) node [black,midway,xshift=8pt] {$r$};          
          \draw (axis cs:1.96371,1.90526) -- (axis cs:0,0) node [black,midway,yshift=8pt] {$r_i$};

          \draw[black] (axis cs:0,0)+(0:1em) arc (0:44.1:1em);
          \node [] at (axis cs:0.65,0.27) {$\theta_i$};

        \end{axis}
      \end{tikzpicture}
      \captionsetup{width=0.9\linewidth}
      \captionof*{figure}{Parameters: $\beta=3$, $\gamma=4/3$. In this case $\DET<0$ and $c>0$.}
  \end{minipage}
  \hspace{4em}
  \begin{minipage}{.4\textwidth}
    \centering
      \begin{tikzpicture}[scale=1,transform shape]
        \begin{axis}[xtick=\empty, ytick=\empty, 
            trig format plots=rad, 
            axis lines = middle, axis line style={->}, axis equal,axis on top,
            xmin=-35, xmax=5, ymin=-20, ymax=20,
            x label style={at={(axis cs:8.5,0)},anchor=north},
            y label style={at={(axis cs:0,20)},anchor=east},            
            xlabel = $\Re$, ylabel=$\Im$]

          \draw[thick, fill = blue!10!white](axis cs:-16.3528,0) circle [draw=black, radius=5.45em];  
          \node [label={270:{$c$}},circle,fill,inner sep=1pt] at (axis cs:{-16.3528},0) {};
          \node [label={-45:{$0$}},circle,fill,inner sep=1pt] at (axis cs:0,0) {};          
          \node [label={120:{$\zeta_1$}},circle,fill,inner sep=1pt] at (axis cs:-2,{2*sqrt(3)}) {};
          \node [label={240:{$\zeta_2$}},circle,fill,inner sep=1pt] at (axis cs:-2,{-2*sqrt(3)}) {};
          \node [circle,fill,inner sep=1pt] at (axis cs:-1.735,2.31523) {};
          \draw (axis cs:-1.735,2.31523) -- (axis cs:{-16.3528},0) node [black,midway,yshift=6pt] {$r$};
          \draw (axis cs:-1.735,2.31523) -- (axis cs:0,0);

          \draw[black] (axis cs:0,0)+(0:0.5em) arc (0:125:0.5em);
          \node [] at (axis cs:-3,1) {$z_i$};
          \node [] at (axis cs:3,2) {$\theta_i$};
        \end{axis}
      \end{tikzpicture} 
      \captionsetup{width=0.9\linewidth}
      \captionof*{figure}{Parameters: $\beta=4$, $\gamma=1/2$. In this case $\DET>0$ and $c<0$.}    
  \end{minipage}
  \caption{Illustrations for the programs \eqref{eqn:program} and \eqref{eqn:program2}.
  Feasible regions are colored blue.}
  \label{fig:feasible}
\end{figure}

The usefulness follows from the next lemma.
Let the optimal value of \eqref{eqn:program} be $\lambdastar{d}$. 
Thus $[0,\lambdastar{d})\cap \prodK{d}=\emptyset$.
The lemma below follows from the fact that the complex plane is a normal Hausdorff space and both $\prodK{d}$ and $[0,\lambda]$ are closed sets for any $\lambda>0$.

\begin{lemma}\label{lem:compliment}
  For any $d\ge 1$ and any $\lambda<\lambdastar{d}$, there is a $\delta$-strip containing $[0,\lambda]$ that does not intersect $\prodK{d}$ for some small $\delta$.
\end{lemma}

It remains to solve the program \eqref{eqn:program}.
Suppose the minimum is achieved by some $\*z=\{r_i e^{\ii\theta_i}\}_{i\in[d]}$.
First assume that there are at least two $z_i$ in the right half plane, say $z_1$ and $z_2$.
In other words, $\theta_i\in\righthalfplane$, for $i=1,2$.
We replace $\theta_1$ and $\theta_2$ by $\theta_1'=\theta_1+\pi\mod2\pi$ and $\theta_2'=\theta_2'+\pi\mod2\pi$.
The effect of this substitution is
\begin{align*}
  \theta_1+\theta_2 & = \theta_1'+\theta_2' \mod 2\pi, & \theta_1',\theta_2'\in\lefthalfplane.
\end{align*}
Moreover, for $i\in\{1,2\}$, if $r_ie^{\ii\theta_i}\in \Complement{\+D(c,r)}$, then $r_ie^{\ii\theta_i'}\in \Complement{\+D(c,r)}$ as well.
This is because that the center of $\+D(c,r)$ is a positive real number.
Therefore, we may assume that there is at most one $z_i$ such that $\theta_i\in\righthalfplane$.

Next observe that if we shrink $r_i$ until $z_i$ is on the circle $\+C(c,r)$, then the optimal value only improves.
Thus we may assume that all $z_i$ are on the circle $\+C(c,r)$.
As a consequence, $r_i$ is determined by $\theta_i$ for all $i\in[d]$.
Indeed, by the cosine law and \eqref{eqn:r},
\begin{align*}  
  r_i^2+c^2-2cr_i\cos \theta_i &= r^2 = \frac{\beta\gamma-1}{\gamma^{2}}+\left( c+\frac{1}{\gamma} \right)^2,
\end{align*}
which implies that
\begin{align*}
  r_i^2-2cr_i\cos \theta_i - \frac{\beta+2c}{\gamma}=0.
\end{align*}
Since one of the solutions is negative, solving $r_i$ we have that $r_i=f(\theta_i)$, where 
\begin{align}\label{eqn:r-theta}
  f(x)\defeq c\cos x+\sqrt{c^2\cos^2x+\frac{\beta+2c}{\gamma}}.
\end{align}

The next lemma states that we can further assume that all $z_i$ on the left half plane to be the same.

\begin{lemma}  \label{lem:Jensen}
  Let $0\le c<r$.
  Suppose all $i\in[k]$, $z_i=r_ie^{\ii\theta_i}$ is on $\+C(c,r)$ and $\theta_i\in\lefthalfplane$.
  Let $\zavg=\ravg e^{\ii\thetaavg}$ be on $\+C(c,r)$ such that $\thetaavg = \frac{1}{k}\sum_{i=1}^k\theta_i$.
  Then, $\prod_{i=1}^k r_i \ge \ravg^k$.
\end{lemma}
\begin{proof}
  We just need to show that if $x\in\lefthalfplane$,
  then $g(x):=\log f(x)$ is a convex function and Jensen's inequality applies,
  where $f(x)$ is defined in \eqref{eqn:r-theta}.
  This can be verified by straightforward calculation that
  \begin{align*}
    g''(x)=-\cos x\cdot\frac{c\sqrt{\gamma}(\beta + 2c + c^2 \gamma)}{\left(\beta +2 c+ c^2 \gamma \cos^2 x\right)^{3/2}}\ge0,
  \end{align*}
  as $x\in\lefthalfplane$.
\end{proof}

We still need to handle the possibility that one of $z_i$, say $z_1$, is on the right half plane.

\begin{lemma}  \label{lem:right-half-plane}
  Let $0\le c < r$.
  Let $d\ge2$ be an integer and $k$ be another integer whose parity is the opposite from that of $d$.
  Let $z_1$ and $\zavg$ be two complex numbers on $\+C(c,r)$.
  Suppose that $z_1=r_1e^{\ii \theta_1}$ where $\theta_1\in\righthalfplane$ and $\zavg=\ravg e^{\ii\thetaavg}$ where $\thetaavg\in\lefthalfplane$.
  If $\theta_1+(d-1)\thetaavg=k \pi$ is fixed, 
  then the minimum of $\ravg^{d-1}r_1$ is attained either when $\theta_1=\pi/2$ or $\theta_1=0$.
\end{lemma}
\begin{proof}
  As $\pi=-\pi\mod 2\pi$, by taking the complex conjugate if necessary, we may assume that $\theta_1\in[0,\pi/2]$.
  Then, as $\theta_1$ increases, $\thetaavg$ decreases.
  If $\thetaavg\in(\pi,3\pi/2]$,
  then as $\theta_1$ increases, both $r_1$ and $\ravg$ decreases and the lemma holds.
  So we only need to handle the case that $\thetaavg\in[\pi/2,\pi]$.

  As $\theta_1+(d-1)\thetaavg=k \pi$, $\thetaavg=\frac{k\pi-\theta_1}{d-1}$.
  Using \eqref{eqn:r-theta}, we then can write $(d-1)\log \ravg+\log r_1$ as a function in $\theta_1$, denoted $\tau(\theta_1)$.
  The minimum of $\ravg^{d-1}r_1$ is attained as long as the minimum of $\tau(\theta_1)$ is attained.
  Straightforward calculation yields
  \begin{align*}
    \tau'(\theta_1) = c\sqrt{\gamma}\left( \frac{\sin(\thetaavg)}{\sqrt{2c+\beta+c^2\gamma\cos^2(\thetaavg)}}
    -\frac{\sin(\theta_1)}{\sqrt{2c+\beta+c^2\gamma\cos^2(\theta_1)}} \right).
  \end{align*}
  Note that $\frac{x}{\sqrt{2c+\beta+c^2\gamma(1-x^2)}}$ is an increasing function for $0<x<1$.

  If $\theta_1+\thetaavg \ge \pi$,
  then $\sin(\thetaavg)\le\sin(\theta_1)$ and $\tau$ is a decreasing function in $\theta_1$.
  In this case, if we increase $\theta_1$, the decrease of $\thetaavg$ is smaller,
  and thus the assumption that $\theta_1+\thetaavg \ge \pi$ is maintained.
  We can keep increasing $\theta_1$ until it hits $\pi/2$.

  Otherwise $\theta_1+\thetaavg < \pi$ and $\tau$ is increasing.
  Similar to the case above, we can keep decreasing $\theta_1$ until it hits $0$.
  The lemma follows from the two cases above.
\end{proof}

Now we can argue when the minimum of the program \eqref{eqn:program} is achieved.

\begin{lemma}  \label{lem:extremal}
  Let $0\le c < r$.
  For any $d\ge 2$,
  the minimum of the program \eqref{eqn:program} is achieved when all $z_i$'s are equal and $\theta_i=\frac{d-1}{d}\cdot\pi$ for all $i$.
\end{lemma}
\begin{proof}
  As argued above, we may assume that either all $z_i$'s are on the left half plane or only $z_1$ is on the right half plane.
  In the former case, by \Cref{lem:Jensen},
  we may assume that all $z_i$'s in the left half plane are equal.
  In the latter case, by \Cref{lem:extremal},
  We can assume that either $\theta_1=\pi/2$ or $\theta_1=0$:
  \begin{itemize}
    \item if $\theta_1=\pi/2$, then we invoke \Cref{lem:Jensen} again to reduce to the case where all $z_i$'s are equal;
    \item if $\theta_1=0$, then by \Cref{lem:Jensen},
      we can assume that all other $z_i$'s are equal.
      As $\pi=-\pi\mod 2\pi$, by taking the complex conjugate if necessary, we can also assume that $\theta_i\in[\pi/2,\pi]$ for all $i\ge 2$.
      Then because of the constraint on $\sum_{i=2}^d\theta_i$,
      there must exist a positive integer $k$ whose parity is opposite to that of $d$ and that $\theta_i=\frac{k\pi}{d-1}$ for all $i\ge 2$.
      It is a simple geometric fact that if $\theta_i\in[\pi/2,\pi]$, $r_i$ decreases as $\theta_i$ increases as $z_i$ lies on $\+C(c,r)$.
      On the other hand, $\theta_i\le \pi$ implies that $k\le d$. 
      Because $k$ has the opposite parity against $d$,
      to achieve the minimum in \eqref{eqn:program}, $k=d-1$ and $\theta_i=\pi$ for all $i\ge2$.

      As $d\ge 2$, consider $r_1r_2$.
      Since $\theta_1=0$ and $\theta_2=\pi$,
      $r_1=r+c$ and $r_2=r-c$, and $r_1r_2=r^2-c^2$.
      We can replace both of them by $z_1'=z_2'=r'e^{\ii\theta'}$ where $r'=\sqrt{r^2-c^2}$ and $\theta'=\frac{\pi}{2}$.
      It is straightforward to verify that $z_1'$ and $z_2'$ are on the circle $\+C(c,r)$, and $r_1r_2=r_1'r_2'$.
      Thus we are reduced to the setting of \Cref{lem:Jensen},
      and applying it makes all $z_i$'s equal.
  \end{itemize}
  To summarize, in all cases, we can assume that all $z_i$'s are equal.

  Similar to the complicated case above,
  we now assume that there is an integer $k$ such that $\theta_i=\frac{k\pi}{d}$ for all $i\in[d]$, $k\le d$,
  and $k$ has the opposite parity against $d$.
  Once again, the larger $k$ the smaller $\prod_{i\in[d]}r_i$.
  Thus, the minimum is achieved when $k=d-1$ and $\theta_i =\frac{d-1}{d}\cdot\pi$ for all $i$.
  The lemma holds.
\end{proof}

By \Cref{lem:extremal} and \eqref{eqn:r-theta},
\begin{align}
  \lambdastar{d}&= \left(f\left( \pi-\frac{\pi}{d} \right)\right)^d =
  \left(c\cos\left( \pi-\frac{\pi}{d}\right)+\sqrt{c^2\cos^2\left(\pi-\frac{\pi}{d}\right)+\frac{\beta+2c}{\gamma}}\right)^d \notag\\
  & = \left(-c\cos\frac{\pi}{d}+\sqrt{c^2\cos^2\frac{\pi}{d}+\frac{\beta+2c}{\gamma}}\right)^d \label{eqn:lambdastard}.
\end{align}

Still, as we want to deal with vertices of all degrees,
we need to determine when the expression in~\eqref{eqn:lambdastard} attains its minimum when $d$ varies.
We will view $\lambdastar{d}$ as a function of $d$ with the expression in \eqref{eqn:lambdastard},
and relax $d$ to be a continuous variable taking values in $[2,\infty)$.
With this in mind, let $h(d)\defeq\log\lambdastar{d}$.
We take derivatives of $h(d)$:
\begin{align}
  h'(d) & = \frac{\log\lambdastar{d}}{d} - \frac{c\pi \sqrt{\gamma}  \sin(\pi/d)}{d \sqrt{\beta + 2c + c^2 \gamma \cos^2(\pi/d)}}; \label{eqn:h'd}\\
  h''(d) & = \cos \frac{\pi}{d}\cdot\frac{c\pi^2\sqrt{\gamma}(\beta + 2c + c^2 \gamma)}{d^3\left(\beta +2 c+ c^2 \gamma \cos^2 (\pi/d )\right)^{3/2}}\label{eqn:h''d}.
\end{align}
As $d\ge 2$, $h''(d)\ge 0$ and $h(d)$ is a convex function.
Thus, the minimum of $h(d)$ (and therefore that of $\lambdastar{d}$) is attained at the solution to $h'(d)=0$.
Call the solution $\dstar$ and let $\lambdastar{}\defeq \lambdastar{\dstar}$.
To summarize the argument above, we have the following lemma.

\begin{lemma}\label{lem:lambdastar}
  Let $0\le c< r$. For any $d\ge 2$, $\lambdastar{d}\ge\lambdastar{}$.
\end{lemma}

The only remaining task is to find out how large $\lambdastar{}$ is and this depends on the value of $c$.
Recall $\zeta_1$ and $\zeta_2$ in \eqref{eqn:zeta}.
We will choose $c$ such that $\pi-\pi/\dstar=\arg{\zeta_1}$.
In other words, we want that $\tan(\pi/\dstar) = \sqrt{\beta\gamma-1}$.
Let $\dstar\defeq\frac{\pi}{\ArcTan {\sqrt{\beta\gamma-1}}}$,
where we take the principle branch of $\ArcTan{\cdot}\in(-\pi/2,\pi/2)$.
In this case, $\lambdastar{} = \abs{\zeta_1}^{\dstar} = \left(\frac{\beta}{\gamma}\right)^{\dstar/2}$.

\begin{lemma}\label{lem:lambdastar-value}
  If $\DET<0$,
  then we can choose $c=\cstar>0$ in \eqref{eqn:c} so that $\lambdastar{}=\left(\frac{\beta}{\gamma}\right)^{\dstar/2}$, 
  where $\dstar = \frac{\pi}{\ArcTan {\sqrt{\beta\gamma-1}}}$.
\end{lemma}
\begin{proof}
  Denote the right hand side of \eqref{eqn:h'd} by $\rho(c,d)$.
  Then
  \begin{align*}
    \rho(c,\dstar) = \log\left(\sqrt{\frac{\beta}{\gamma}}\right) - \ArcTan {\sqrt{\beta\gamma-1}}\cdot\sqrt{\beta\gamma-1} \cdot \frac{c}{\beta +c}.
  \end{align*}
  It is straightforward to verify that $\rho(\cstar,\dstar)=0$.

  Since $h''(d)\ge 0$, $h'(d)=0$ has at most one zero in $d$ for any fixed $c$.
  Once we chose $c=\cstar$, $\dstar$ is the unique zero of $h'(d)$.
  The lemma follows.
\end{proof}

\subsection{\texorpdfstring{$\DET>0$}{Phi > 0}}
\label{sec:phi>0}

When $\DET>0$, the argument is almost the same as or even simpler than that in \Cref{sec:phi<0}.
The main issue is that following \Cref{lem:lambdastar-value} would yield $c<0$ and some geometry changes.
We consider instead a disk $\+D(c,r)$ with $c<-\beta<0$.
Eventually we also choose $c=\cstar$ according to \eqref{eqn:c},
although now $\cstar<-\beta<0$ as $\Phi>0$.
The radius $r$ is still chosen according to \eqref{eqn:r} such that $\zeta_1,\zeta_2$ are on $\+C(c,r)$.
The main change is that now we choose region $K=\Closure{\+D(c,r)}$.
Namely, $K$ is the closure of $\+D(c,r)$ instead of its complement.
An illustration of $K$ and $\+K_i$ for $i=2,3,4,5$ is given in \Cref{fig:contraction-negative}.

\begin{figure}[htpb]
  \centering
  \begin{minipage}{.15\linewidth}
    \includegraphics[width=\linewidth]{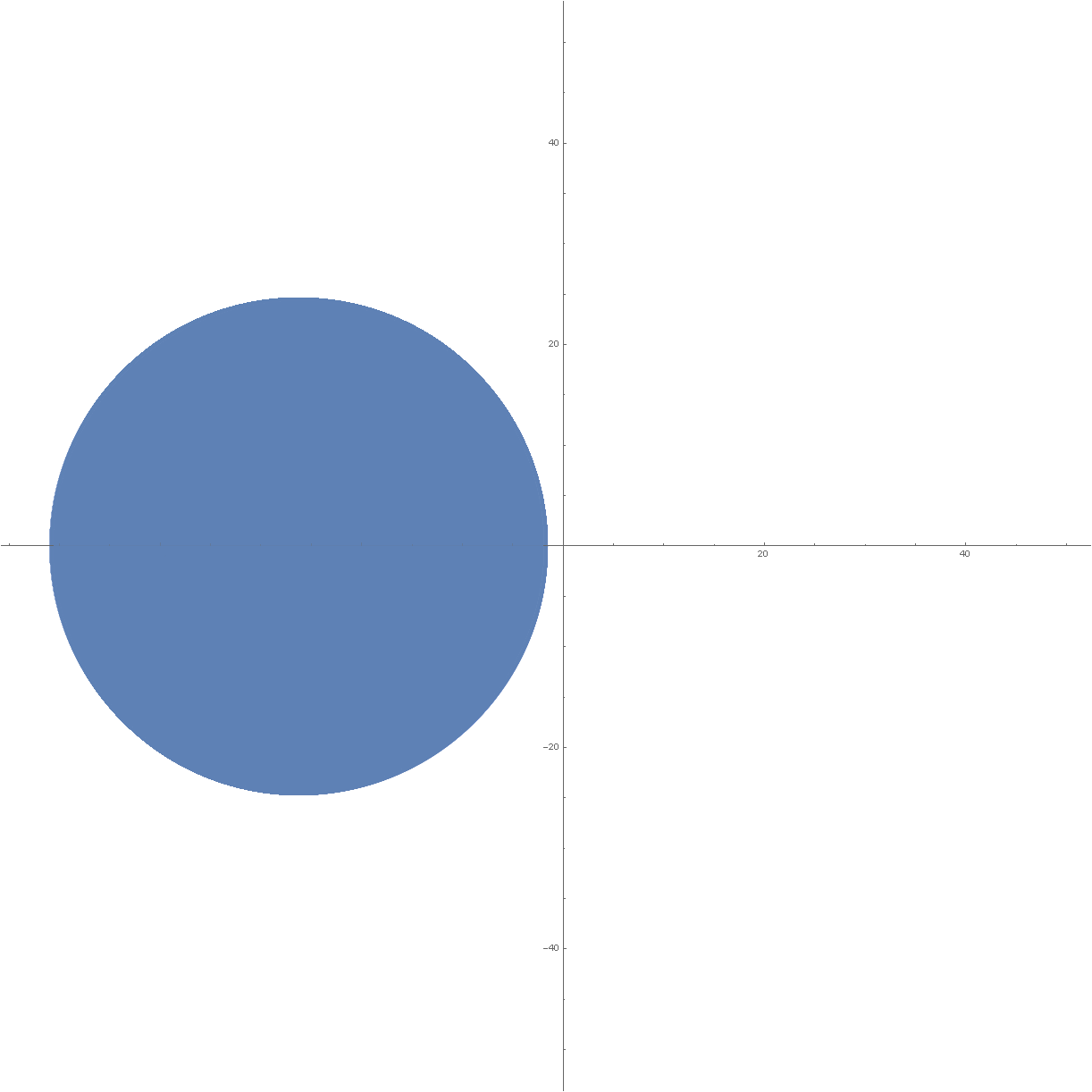}
    \captionsetup{width=0.9\linewidth}
    \captionof*{figure}{$K=\Closure{\+D(c,r)}$}
  \end{minipage}
  \hspace{0.5em}
  \begin{minipage}{.15\linewidth}
    \includegraphics[width=\linewidth]{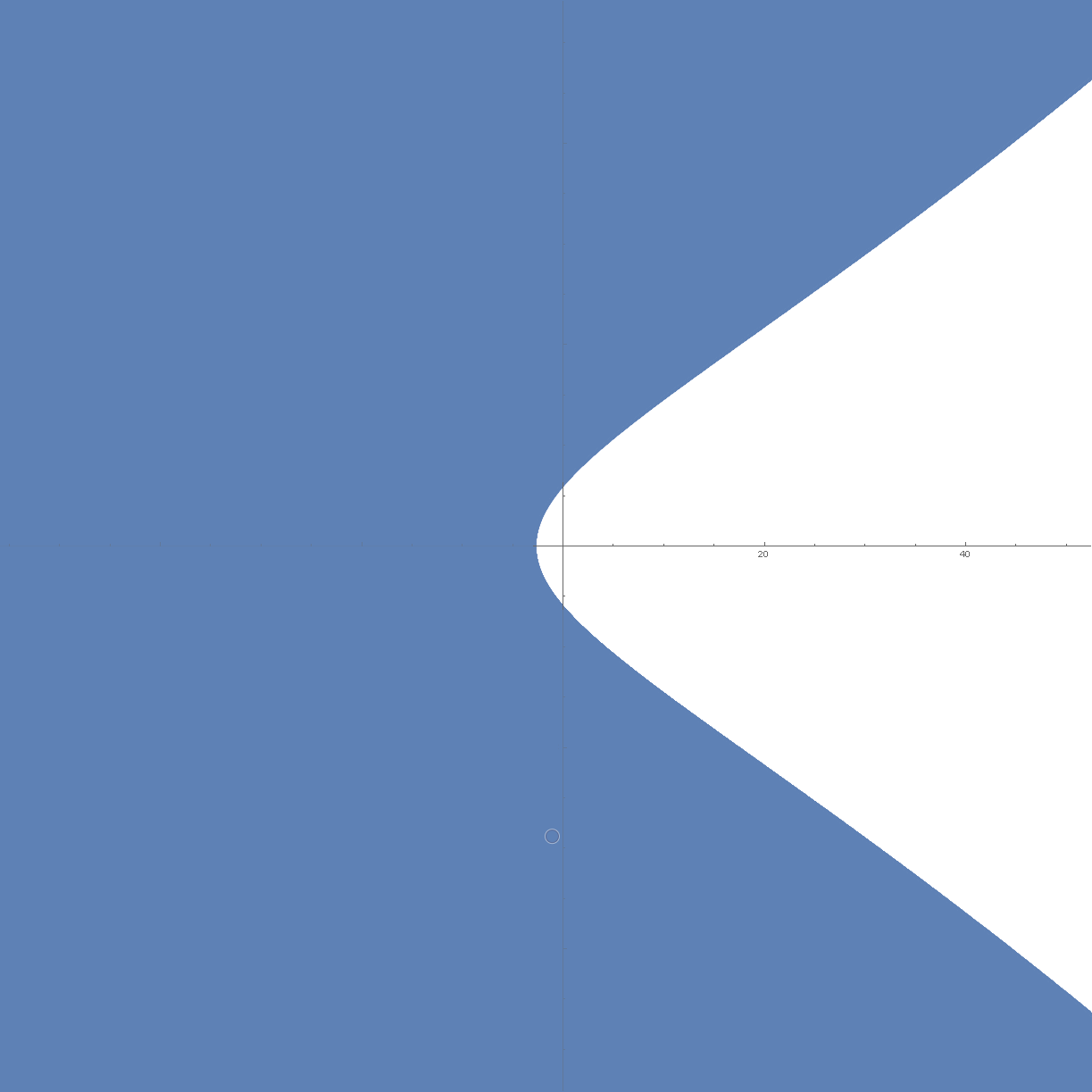}
    \captionsetup{width=0.9\linewidth}
    \captionof*{figure}{$\+K_2=-K\cdot K$}
  \end{minipage}
  \hspace{0.5em}
  \begin{minipage}{.15\linewidth}
    \includegraphics[width=\linewidth]{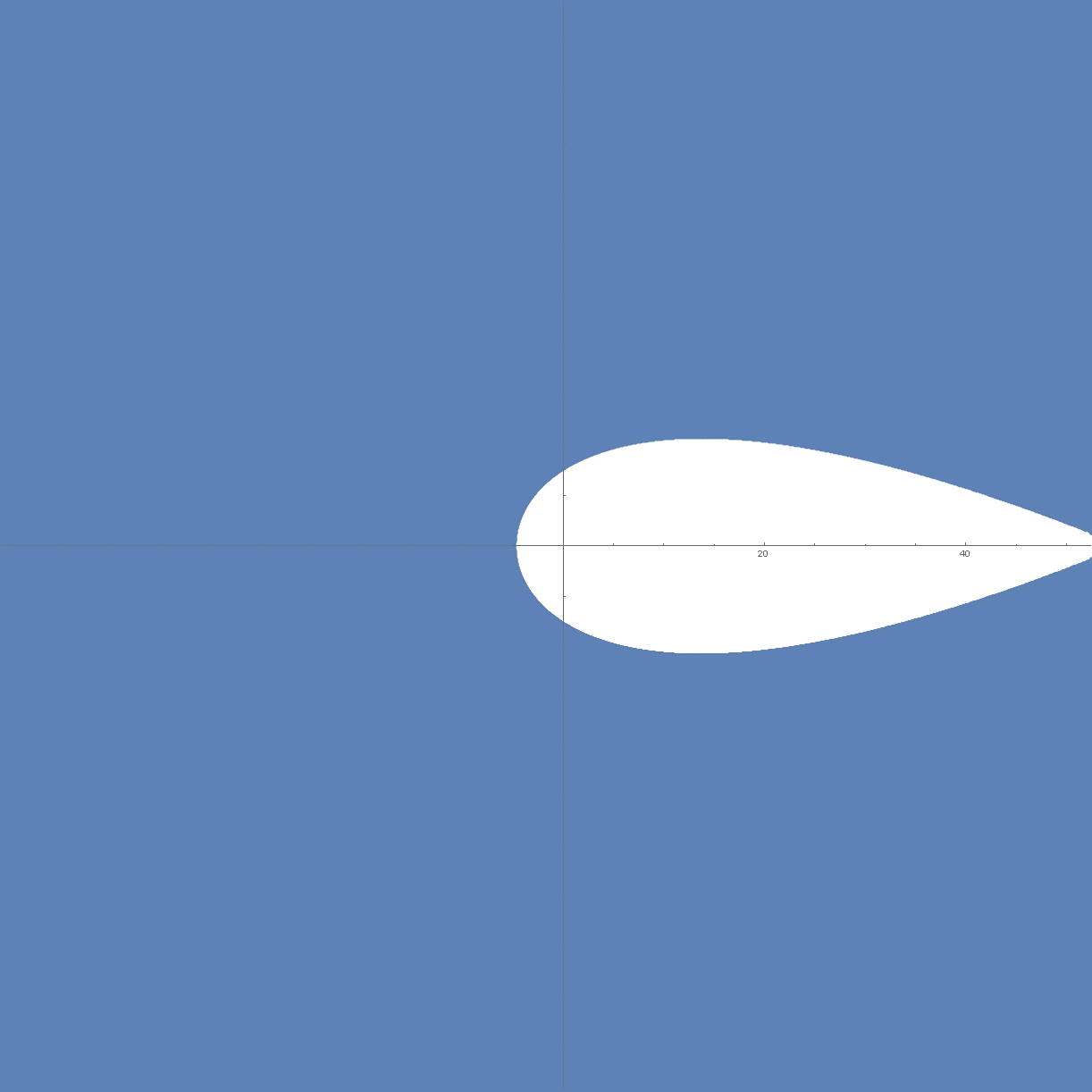}
    \captionsetup{width=0.9\linewidth}
    \captionof*{figure}{$\+K_3= K^3$}
  \end{minipage}
  \hspace{0.5em}
  \begin{minipage}{.15\linewidth}
    \includegraphics[width=\linewidth]{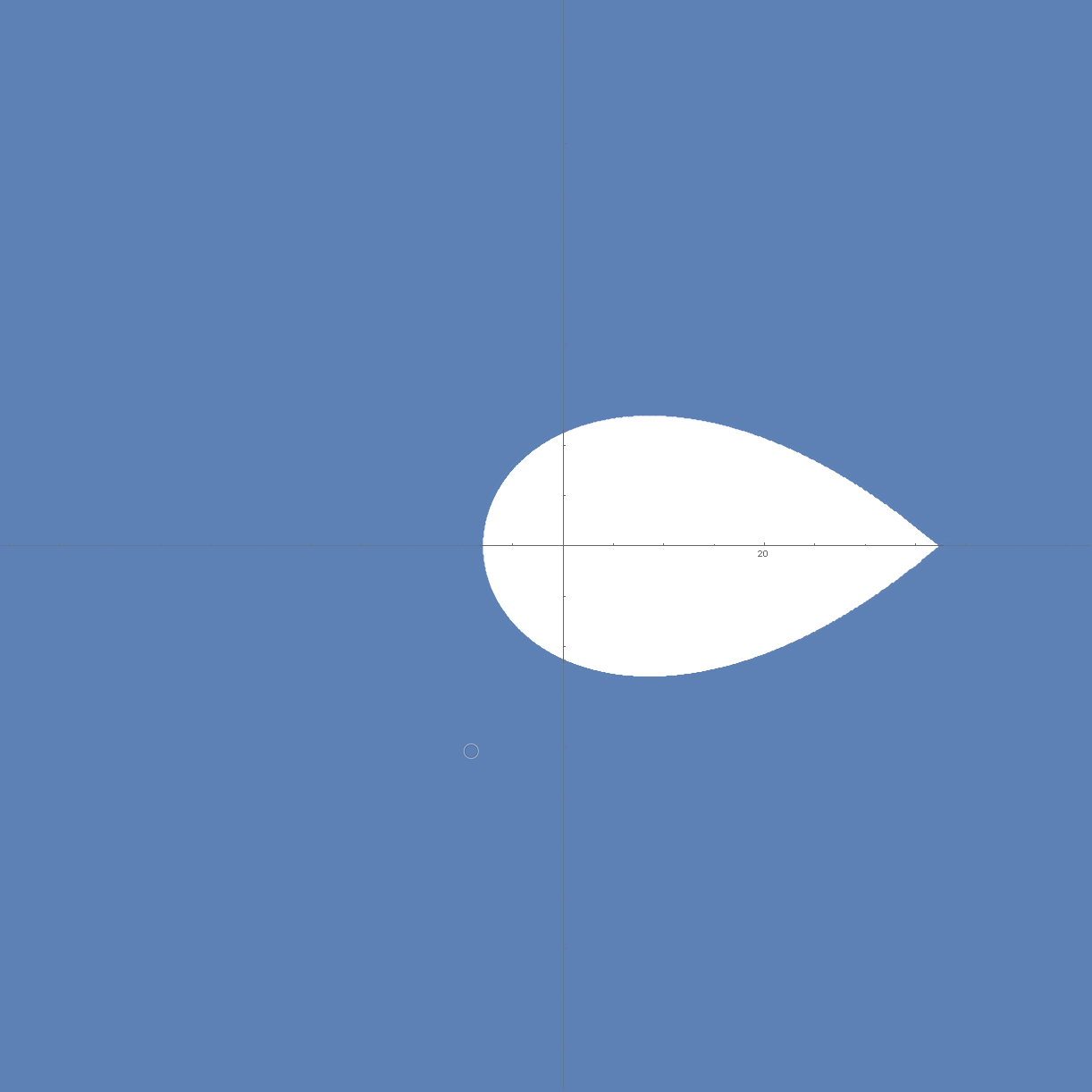}
    \captionsetup{width=0.9\linewidth}
    \captionof*{figure}{$\+K_4= -K^4$}
  \end{minipage}
  \hspace{0.5em}
  \begin{minipage}{.15\linewidth}
    \includegraphics[width=\linewidth]{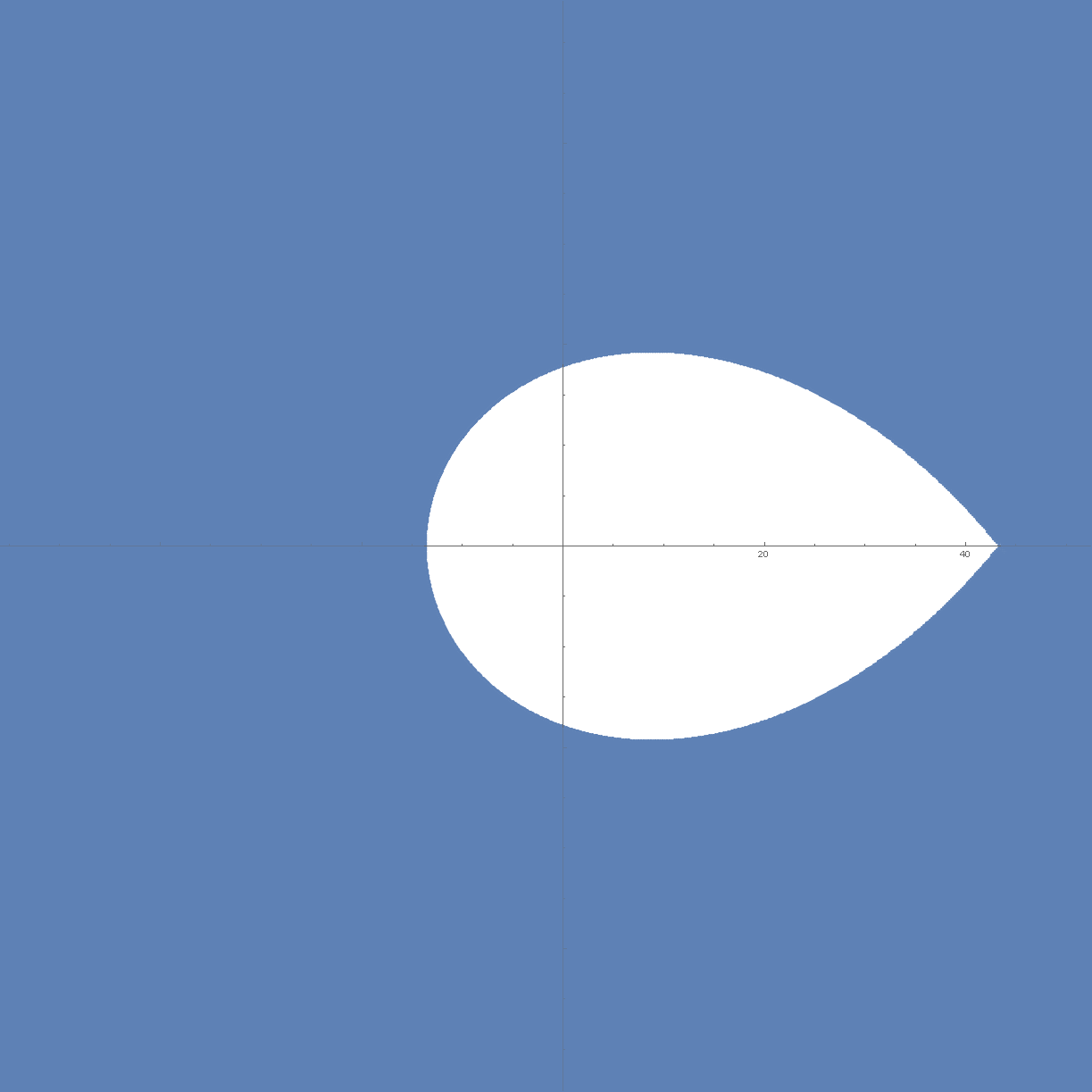}
    \captionsetup{width=0.9\linewidth}
    \captionof*{figure}{$\+K_5=K^5$}
  \end{minipage}
  \caption{Our region $K=\Closure{\+D(c,r)}$, $\+K_2$, $\+K_3$, $\+K_4$ and $\+K_5$ in the case of $\beta=4$ and $\gamma=\frac 1 2$.
  Here, the intercept of $\+K_d$ on the positive real line is minimised at $d=4$ for all $d\ge 2$.}
  \label{fig:contraction-negative}
\end{figure}
Then, the program \eqref{eqn:program} becomes
\begin{align}
  \label{eqn:program2}
  \begin{split}
  \min\quad & \prod_{i=1}^d r_i \\
  \text{subject to}\quad & \sum_{i=1}^d\theta_i = 
  \begin{cases}
    0 \mod 2\pi & \text{if $d$ is odd;}\\
    \pi \mod 2\pi & \text{if $d$ is even;}\\
  \end{cases}\\
  \quad & \forall i\in[d],\ r_i\ge 0 \text{ and } 0 \le \theta_i < 2\pi;\\
  \quad & \forall i\in[d],\ \abs{r_ie^{\ii\theta_i}-c}\le r.    
  \end{split}
\end{align}
Still denote the optima of \eqref{eqn:program2} by $\lambdastar{d}$ and it is easy to verify that \Cref{lem:compliment} holds in this setting.
An illustration can be found in \Cref{fig:feasible}.

As $c<-\beta$, it is easy to verify that $c<-r$ using $\eqref{eqn:r}$, and $0\not\in K$.
So for any $z_i$, we can shrink it until it hits the right boundary of $\+C(c,r)$.
In this case, similar to \eqref{eqn:r-theta}, $r_i=\alternative{f}(\theta_i)$, where
\begin{align}  \label{eqn:r-theta2}
  \alternative{f}(x)\defeq c\cos x-\sqrt{c^2\cos^2x+\frac{\beta+2c}{\gamma}}.
\end{align}
The sign changed because now there are two positive solutions and we should choose the smaller one.
Moreover, notice that due to the constraints in \eqref{eqn:program2},
$\theta_i\in\lefthalfplane$ and is further constrained into a range so that $\alternative{f}(\cdot)$ is real, namely
\begin{align}  \label{eqn:theta-range}
  c^2\gamma\cos^2\theta_i+\beta+2c>0.
\end{align}
In particular, since $\theta_i=\pi$ satisfies the constraint of \eqref{eqn:program2},
$c^2\gamma+\beta+2c>0$.
The analogue of \Cref{lem:Jensen} also holds.

\begin{lemma}  \label{lem:Jensen2}
  Let $c<-r<0$.
  Suppose all $i\in[k]$, $z_i=r_ie^{\ii\theta_i}$ where $r_i=\alternative{f}(\theta_i)$ and \eqref{eqn:theta-range} holds.
  Let $\zavg=\ravg e^{\ii\thetaavg}$ be on $\+C(c,r)$ such that $\thetaavg = \frac{1}{k}\sum_{i=1}^k\theta_i$.
  Then, $\prod_{i=1}^k r_i \ge \ravg^k$.
\end{lemma}
\begin{proof}
  The proof goes through similar calculations to that of \Cref{lem:Jensen}.
  Let $\alternative{g}=\log\alternative{f}$.
  Then for $x$ such that \eqref{eqn:theta-range} holds,
  \begin{align*}
    \alternative{g}''(x)=\cos x\cdot\frac{c\sqrt{\gamma}(\beta + 2c + c^2 \gamma)}{\left(\beta +2 c+ c^2 \gamma \cos^2 x\right)^{3/2}}\ge0,
  \end{align*}
  as $\beta + 2c + c^2 \gamma>0$.
\end{proof}

Since in this case all $z_i$'s are on the left half plane, there is no need to consider $z_i$'s on the right half plane like \Cref{lem:right-half-plane}.
We directly go to the analogue of \Cref{lem:extremal}.

\begin{lemma}  \label{lem:extremal2}
  Let $c < -r<0$.
  For any $d\ge 2$,
  the minimum of the program \eqref{eqn:program2} is achieved when all $z_i$'s are equal and $\theta_i=\frac{d-1}{d}\cdot\pi$ for all $i$.
\end{lemma}
\begin{proof}
  We first invoke \Cref{lem:Jensen2} to assume that all $z_i$'s are equal.
  Therefore there exists $k$ of opposite parity against $d$ such that $\theta_i=\frac{k\pi}{d}$.
  We may assume that $\theta_i\in[\pi/2,\pi]$ by taking conjugates if necessary.
  Then, $r_i$ is a decreasing function in $\theta_i$, and the minimum of $\prod_{i=1}^dr_i$ is achieved when $k=d-1$.
\end{proof}

Some calculations need to be changed due to the sign change in \eqref{eqn:r-theta2}.
By \Cref{lem:extremal2} and \eqref{eqn:r-theta2},
\begin{align}
  \lambdastar{d}& = \left(-c\cos\frac{\pi}{d}-\sqrt{c^2\cos^2\frac{\pi}{d}+\frac{\beta+2c}{\gamma}}\right)^d \label{eqn:lambdastard2}.
\end{align}

Let $\alternative{h}(d)\defeq\log\lambdastar{d}$ where $\lambdastar{d}$ is given as the expression in \eqref{eqn:lambdastard2}.
We take derivatives of $\alternative{h}(d)$:
\begin{align}
  \alternative{h}'(d) & = \frac{\log\lambdastar{d}}{d} + \frac{c\pi \sqrt{\gamma}  \sin(\pi/d)}{d \sqrt{\beta + 2c + c^2 \gamma \cos^2(\pi/d)}}; \label{eqn:h'd2}\\
  \alternative{h}''(d) & = -\cos \frac{\pi}{d}\cdot\frac{c\pi^2\sqrt{\gamma}(\beta + 2c + c^2 \gamma)}{d^3\left(\beta +2 c+ c^2 \gamma \cos^2 (\pi/d )\right)^{3/2}}\label{eqn:h''d2}.
\end{align}
As $d\ge 2$ and $\beta + 2c + c^2 \gamma>0$, $\alternative{h}''(d)\ge 0$ and $\alternative{h}(d)$ is still a convex function.
Thus, the minimum of $\alternative{h}(d)$ is attained at the solution to $\alternative{h}'(d)=0$.
With a little abuse of notation, call the solution $\dstar$ and let $\lambdastar{}\defeq \lambdastar{\dstar}$.

\begin{lemma}\label{lem:lambdastar2}
  Let $c<-r<0$. For any $d\ge 2$, $\lambdastar{d}\ge\lambdastar{}$.
\end{lemma}

We still need to choose $c$ so that $\dstar=\frac{\pi}{\ArcTan {\sqrt{\beta\gamma-1}}}$.

\begin{lemma}\label{lem:lambdastar-value2}
  If $\DET>0$,
  then we can choose $c=\cstar<-\beta<0$ in \eqref{eqn:c} so that $\lambdastar{}=\left(\frac{\beta}{\gamma}\right)^{\dstar/2}$, 
  where $\dstar = \frac{\pi}{\ArcTan {\sqrt{\beta\gamma-1}}}$.
\end{lemma}
\begin{proof}
  Denote the right hand side of \eqref{eqn:h'd2} by $\alternative{\rho}(c,d)$.
  Then
  \begin{align*}
    \alternative{\rho}(c,\dstar) & = \log\left(\sqrt{\frac{\beta}{\gamma}}\right) + \frac{\pi}{\dstar}\cdot\sqrt{\beta\gamma-1} \cdot \frac{c}{\abs{\beta +c}}\\
    & = \log\left(\sqrt{\frac{\beta}{\gamma}}\right) - \ArcTan {\sqrt{\beta\gamma-1}}\cdot\sqrt{\beta\gamma-1} \cdot \frac{c}{\beta +c}.
  \end{align*}
  It is straightforward to verify that $\rho(\cstar,\dstar)=0$.  

  Since $\alternative{h}''(d)\ge 0$, $\alternative{h}'(d)=0$ has at most one zero in $d$ for any fixed $c$.
  Once we chose $c=\cstar$, $\dstar$ is the unique zero of $\alternative{h}'(d)$.
  The lemma follows.
\end{proof}

\subsection{\texorpdfstring{$\DET=0$}{Phi = 0}}
In fact, the arguments in \Cref{sec:phi<0} and \Cref{sec:phi>0} can be viewed as moving $c$ from $0$ to $\infty$, then ``wrapping around'' to $-\infty$, and eventually to $-\beta$.
The threshold case of $\DET=0$ requires us to take $c=\infty$, in which case $K$ becomes the close half plane $\{z\mid\Re z\le-\frac{1}{\gamma}\}$.
The program becomes 
\begin{align}
  \label{eqn:program3}
  \begin{split}
  \min\quad & \prod_{i=1}^d r_i \\
  \text{subject to}\quad & \sum_{i=1}^d\theta_i = 
  \begin{cases}
    0 \mod 2\pi & \text{if $d$ is odd;}\\
    \pi \mod 2\pi & \text{if $d$ is even;}\\
  \end{cases}\\
  \quad & \forall i\in[d],\ r_i\ge 0 \text{ and }\theta_i\in(\pi/2,3\pi/2);\\
  \quad & \forall i\in[d],\ -r_i\le \frac{1}{\gamma}.    
  \end{split}
\end{align}
Still denote the optima of \eqref{eqn:program3} by $\lambdastar{d}$ and it is easy to verify that \Cref{lem:compliment} holds in this setting.


Once again, we can assume that all $z_i's$ are on the boundary, namely that $\Re z =-\frac{1}{\gamma}$.
In this case
\begin{align}  \label{eqn:r-theta3}
  r_i = -\frac{1}{\gamma\cos\theta_i}.
\end{align}
It is easy to check that $\log r_i = - \log (-\cos\theta_i) -\cos\gamma$ is a convex function.
By the same argument as in \Cref{lem:extremal2},
\begin{align}\label{eqn:lambdastard3}
  \lambdastar{d}= \frac{1}{\gamma^d\cos^d(\pi/d)}.
\end{align}

\begin{lemma}\label{lem:lambdastar-value3}
  If $\DET=0$,
  then choosing $K=\{z\mid\Re z\le-\frac{1}{\gamma}\}$ ensures that for any $d\ge 2$, $\lambdastar{d}\ge \lambdastar{}=\left(\frac{\beta}{\gamma}\right)^{\dstar/2}$, 
  where $\dstar = \frac{\pi}{\ArcTan {\sqrt{\beta\gamma-1}}}$.
\end{lemma}
\begin{proof}
  We just need to verify that $\log\lambdastar{d} = -d\log\gamma -d\log\cos(\pi/d)$ (with $\lambdastar{d}$ in \eqref{eqn:lambdastard3}) takes its minimum at $d=\dstar$.
  In this case,
  \begin{align*}
    \left(\log\lambdastar{d}\right)' & = -\log\gamma-\log\cos(\pi/d)-\frac{\pi\tan(\pi/d)}{d};\\
    \left(\log\lambdastar{d}\right)'' & = \frac{\pi^2}{d^3\cos^2(\pi/d)} \ge 0.
  \end{align*}
  The lemma follows from $\left(\log\lambdastar{\dstar}\right)'=0$,
  which can be verified using $\DET=0$, $\cos(\pi/\dstar)=1/\sqrt{\beta\gamma}$, and $\tan(\pi/\dstar)=\sqrt{\beta\gamma-1}$.
\end{proof}

\subsection{Proof of \texorpdfstring{\Cref{thm:main} and \Cref{thm:zero-free}}{Theorem 2 and Theorem 3}}

%

In order to avoid considering infinitely many degrees, 
we observe the following.

\begin{lemma}  \label{lem:norm>1}
  Let $\beta>\gamma$ be the parameters.
  For any of our chosen $K$, if $z\in K$, then $\abs{z}>1$.
\end{lemma}
\begin{proof}
  The assumption $\beta>\gamma$ implies that $\lambdastar{}>1$.
  Assume otherwise that $\exists z\in K$ such that $\abs{z}\le 1$.
  Then, there must exist some $z'$ close to $z$ so that $\abs{z}=\abs{z'}$, $\arg(z')=\frac{p\pi}{q}$ where $p$ and $q$ are two integers, $p$ is odd and $q$ is even.
  Thus, $-(-z')^q \in \+K_{q}$ and $arg(-(-z')^q)=0$.
  Moreover, $\abs{-(-z')^q} = \abs{z'}^q \le 1 < \lambdastar{}$.
  This contradicts to \Cref{lem:lambdastar-value}, \ref{lem:lambdastar-value2}, or \ref{lem:lambdastar-value3}.
\end{proof}

Our method in fact shows a multivariate version of \Cref{thm:zero-free}.
Recall the definition of the multivariate partition function in \eqref{eqn:Z-multi}.

\begin{theorem}  \label{thm:zero-free-multi}
  Let $\beta,\gamma$ be positive parameters such that $\beta\ge\gamma$ and $\beta\gamma>1$,
  and $\lambdastar{}$ defined as in \Cref{thm:main}.
  There exists a $\delta>0$ such that for any $\lambda'<\lambdastar{}$ and any graph $G=(V,E)$ such that $\deg_G(v)\ge 2$ for all $v\in V$, 
  $\Zspin(G;\vec{\lambda})$ does not vanish in a $\delta$-polystrip containing the poly-region $[0,\lambda']^n$ where $n=\abs{V}$.
\end{theorem}

\begin{proof}
  First we claim that for any $\lambda'<\lambdastar{}$, we can choose a $\delta$-strip $\+N$ containing $[0,\lambda']$ for $\lambda'<\lambdastar{}$ 
  so that it does not intersect $\+K_d$ for any $d\ge 2$,
  and the $\delta$-polystrip $\+N^n$ is what we choose in the theorem.
  \Cref{lem:compliment}, \ref{lem:lambdastar}, \ref{lem:lambdastar-value}, \ref{lem:lambdastar2}, \ref{lem:lambdastar-value2} and \ref{lem:lambdastar-value3}
  together imply that for any single $d$, there is a $\delta_d$-strip covering $[0,\lambda']$ that does not intersect $\+K_d$.
  If $\beta>\gamma$, by \Cref{lem:norm>1}, $\abs{z}>1$ for any $z\in K$.
  For sufficiently large $d$, for any $z\in\+K_d$, $\abs{z}>\lambdastar{}$.
  Thus, we only need to take $\delta$ to be the minimum one among finitely many $\delta_d$'s.
  If $\beta=\gamma$, then $K$ is the unit circle, $\+K_d=K$ for any $d\ge 2$, and $\lambdastar{}=1$.
  In this case, clearly the claim holds.

  We construct a sequence of graphs $G_0,G_1,\dots,G_n=G$.
  In $G_0=(V_0,E_0)$, we replace each vertex $v\in V$ by $d=\deg_G(v)$ copies, denoted $v_1,v_2,\dots,v_d$,
  and connect them according to $E$ so that $G_0$ is a disjoint union of isolated edges.
  Then
  \begin{align*}
    \Zspin(G_0;\lambda) = \prod_{e\in E}\left( \gamma\lambda^2+2\lambda+\beta  \right).
  \end{align*}
  The only zeros of $\Zspin$ are $\zeta_1$ and $\zeta_2$,
  both of which are in the circular region $K$ chosen according to \Cref{lem:lambdastar-value}, \ref{lem:lambdastar-value2}, or \ref{lem:lambdastar-value3}.
  Now consider the polynomial
  \begin{align*}
    \Zspin(G_0;\*z)\defeq\prod_{(u_i,v_j)\in E_0}\left( \gamma z_{u_i} z_{v_j}+z_{u_i}+z_{v_j}+\beta  \right).
  \end{align*}
  By \Cref{prop:Grace-Walsh-Szego},
  $\Zspin(G_0;\*z)$ does not vanish if $z_v\not\in K$ for all $v\in V_0$.
  
  The graph $G_1$ is constructed by choosing an arbitrary $v\in V$,
  and contract $v_1,\dots, v_d$ where $d=\deg_G(v)$.
  Namely, we replace all $\{z_{v_i}\}$ by the same $z_v$ in $\Zspin(G_0;\*z)$ to get a new polynomial $\Zspin(G_1;\*z)$.
  This is the operation in \Cref{lem:contraction},
  and then $\Zspin(G_1;\*z)\neq 0$ if $z_v\not\in \+K_d$ and $z\not\in K$ for all $z\neq z_v$.
  We keep contracting vertices to construct $G_2,\dots,G_n=G$, and their partition functions correspondingly.
  Then \Cref{lem:contraction} guarantees that $\Zspin(G;\*z)\neq 0$ as long as $z_v\not\in \+K_d$ where $d=\deg_G(v)$ for all $v\in V$.

  Our choice of $\+N$ already ensures that any $\+N\cap\+K_d=\emptyset$ for all $d\ge 2$.
%
%
  The theorem follows. 
\end{proof}

\Cref{thm:zero-free} is a simple corollary of \Cref{thm:zero-free-multi}.
To prove \Cref{thm:main}, we need to take some special care of degree $1$ vertices.

\begin{proof}[Proof of \Cref{thm:main}]
  Let $G=(V,E)$ be a graph and $v\in V$ such that $\deg_G(v)=1$.
  Let $\vec{\lambda}$ be the (not necessarily uniform) vertex weights.
  Let the unique neighbour of $v$ in $G$ be $u$.
  The ``pruning'' operation is the following.
  Construct $G'=G[V\setminus\{v\}]$ and $\lambda_w'=\lambda_w$ if $w\neq u$ and $\lambda_u'=\lambda_u\cdot\frac{\lambda_v\gamma+1}{\lambda_v+\beta}$.
  Then $\Zspin(G;\vec{\lambda})=(\beta+\lambda_v) \cdot \Zspin(G';\vec{\lambda}')$.

  Notice that if $\lambda_v\le \frac{\beta-1}{\gamma-1}$, then $\frac{\lambda_v\gamma+1}{\lambda_v+\beta}\le 1$.
  Moreover, $\lambdastar{}\le \lambda_c$ and $\lambda_c\le \frac{\beta-1}{\gamma-1}$ \cite[Lemma 3.2]{GL18}.
  Thus, in the assumed range of parameters, we can keep pruning leaves until there is none,
  and all $\lambda_v$ after pruning still satisfies that $\lambda_v<\lambda'<\lambdastar{}$.
  When there are no degree $1$ vertices, 
  we apply \Cref{lem:zeros-ferro-2-spin-multi} and \Cref{thm:zero-free-multi}.
\end{proof}

\section{Concluding remarks}
\label{sec:limit}

The main limit of our approach is that the roots $\zeta_1,\zeta_2$ to the single edge case are fixed.
Any circular region we choose in \Cref{prop:Grace-Walsh-Szego} and subsequently in \Cref{lem:contraction} must contain $\zeta_1$ and $\zeta_2$.
If the degree $d$ of a vertex is very close to $\ArcTan{\sqrt{\beta\gamma-1}}$,
then $\zeta_1$ will be mapped to very close to the real axis after the contraction.
Thus, our best hope is to make sure that this is the worst case, 
and that is exactly what we do in \Cref{lem:lambdastar-value,lem:lambdastar-value2,lem:lambdastar-value3}.
This seems to be an inherent difficulty to the contraction method on ferromagnetic $2$-spin systems.

\section*{Acknowledgement}
We would like to thank the organisers of the workshop ``Deterministic Counting, Probability, and Zeros of Partition Functions'' in the Simons Institute for the Theory of Computing.
The topic of the workshop inspired us to look at this problem and the work was initiated during the workshop.
HG also wants to thank the hospitality of the Institute of Theoretical Computer Science in Shanghai University of Finance and Economics,
where part of the work was done.

\bibliographystyle{alpha} 
\bibliography{zeros}

\end{document}